\documentclass[11pt,english]{article}
\usepackage[utf8]{inputenc}
\usepackage[T1]{fontenc}
\usepackage{lmodern}
\usepackage[DIV=15]{typearea} % large value = less margin (roughly speaking), recommended values: 10pt: 8, 11pt: 10, 12pt: 12

\usepackage{verbatim}
\usepackage{float}
\usepackage{array}
\usepackage{amsmath}
\usepackage{amsthm}
\usepackage{amssymb}
\usepackage{mathtools,amsfonts,enumerate,enumitem}
\usepackage{thmtools, thm-restate}
\makeatletter

\floatstyle{ruled}
\newfloat{algorithm}{tbp}{loa}
\providecommand{\algorithmname}{Algorithm}
\floatname{algorithm}{\protect\algorithmname}

\usepackage{xcolor}
\usepackage{amsthm, thmtools}
\usepackage{nameref}
\definecolor{ForestGreen}{rgb}{0.1333,0.5451,0.1333}
\definecolor{DarkRed}{rgb}{0.8,0,0}
\definecolor{Red}{rgb}{1,0,0}
\usepackage[linktocpage=true,
pagebackref=true,colorlinks,
linkcolor=DarkRed,citecolor=ForestGreen,
bookmarks,bookmarksopen,bookmarksnumbered]
{hyperref}

\usepackage{cleveref}

\declaretheorem[numberwithin=section,refname={Theorem,Theorems},Refname={Theorem,Theorems},name={Theorem}]{thm}

\declaretheorem[numberlike=thm,refname={Lemma,Lemmas},Refname={Lemma,Lemmas},name={Lemma}]{lem}
\declaretheorem[numberlike=thm,refname={Lemma,Lemmas},Refname={Lemma,Lemmas},name={Lemma}]{lemma}
\declaretheorem[numberlike=thm,refname={Corollary,Corollaries},Refname={Corollary,Corollaries},name={Corollary}]{cor}

\declaretheorem[numberlike=thm,refname={Proposition,Propositions},Refname={Proposition,Propositions},name={Proposition}]{prop}
\declaretheorem[numberlike=thm,refname={Observation,Observations},Refname={Observation,Observations}]{observation}

\declaretheorem[numberlike=thm,refname={Definition,Definitions},Refname={Definition,Definitions},name={Definition}]{definition}
\declaretheorem[style=remark,numberwithin=section,refname={Remark,Remarks},Refname={Remark,Remarks},name={Remark}]{rem}
\declaretheorem[style=remark,numberlike=thm,refname={Claim,Claims},Refname={Claim,Claims}]{claim}

\AtBeginDocument{%
\let\ref\Cref
}

\makeatother

\usepackage[titletoc,title]{appendix}

\ifdefined\ShowComment

\def\selfnote#1{{\bf Self Note:} #1}
\def\thatchaphol#1{\marginpar{$\leftarrow$\fbox{T}}\footnote{$\Rightarrow$~{\sf #1 --Thatchaphol}}}

\else

\def\selfnote#1{}
\def\thatchaphol#1{}

\fi

\begin{document}
\global\long\def\A{\mathcal{A}}
\global\long\def\G{\mathcal{G}}

\newcommand{\first}{\operatorname{first}}
\newcommand{\last}{\operatorname{last}}
\newcommand{\current}{\operatorname{current}}
\newcommand{\vol}{{\hbox{\bf vol}}}
\newcommand{\ex}{{\hbox{ex}}}
\newcommand{\ab}{{\hbox{ab}}}
\global\long\def\uf{\text{\emph{Unit-Flow}}}

\newcommand{\unitflow}{{\em Unit-Flow}\space}

\def\norm#1{\left\| #1 \right\|}
\newcommand{\E}{\mathbb{E}}

\newcommand{\tab}{.\hskip.1in}
\DeclarePairedDelimiter\ceil{\lceil}{\rceil}
\title{Expander Decomposition and Pruning: \\Faster, Stronger, and Simpler.}

\newcommand{\email}[1]{#1}

\author{%
\normalsize
Thatchaphol Saranurak \thanks{Toyota Technological Institute at Chicago, 
\protect\email{saranurak@ttic.edu}. Work partly done while was in KTH Royal Institute of Technology}
\and 
\normalsize
Di Wang \thanks{Georgia Institute of Technology, \protect\email{wangd@eecs.berkeley.edu}.}
}

\maketitle

\pagenumbering{gobble}
\begin{abstract}
We study the problem of graph clustering where the goal is to partition a graph into clusters, i.e. disjoint subsets of vertices, such that each cluster is well connected internally while sparsely connected to the rest of the graph. In particular, we use a natural bicriteria notion motivated by Kannan, Vempala, and Vetta~\cite{KannanVV00} which we refer to as {\em expander decomposition}. Expander decomposition has become one of the building blocks in the design of fast graph algorithms, most notably in the nearly linear time Laplacian solver by Spielman and Teng~\cite{SpielmanT04}, and it also has wide applications in practice.

We design algorithm for the parametrized version of expander decomposition, where given 
a graph $G$ of $m$ edges and a parameter $\phi$, our algorithm finds a partition of the vertices into clusters
such that each cluster induces a subgraph of conductance at least
$\phi$ (i.e. a $\phi$ expander), and only a $\widetilde{O}(\phi)$ fraction of the edges in $G$ have endpoints across different clusters. Our algorithm runs in $\widetilde{O}(m/\phi)$ time, and is the  
first	nearly linear time algorithm when $\phi$ is at least $1/\log^{O(1)} m$, which is the case in most practical settings and theoretical applications. %This affirmatively answers the open question noted in {[}Spielman and Teng, STOC'04{]} and also in {[}Koutis and Miller, SPAA'08; Orecchia and Vishnoi, SODA'11; Orecchia, Sachdeva, and Vishnoi, STOC'12{]}. 
Previous results either take $\Omega(m^{1+o(1)})$ time (e.g. \cite{NanongkaiS16,Wulff-Nilsen16a}), or attain nearly linear time but with a weaker expansion guarantee where each output cluster is guaranteed to be contained inside some unknown $\phi$ expander (e.g. \cite{SpielmanT13,AndersenCL06}). Our result achieve both nearly linear running time and the strong expander guarantee for clusters. Moreover, a main technique we develop for our result can be applied to obtain a much better \emph{expander pruning} algorithm, which is the key tool for maintaining an expander decomposition on dynamic graphs. Finally, we note that our algorithm is developed from first principles based on relatively simple and basic techniques, thus making it very likely to be practical.  
\end{abstract}

\newpage
\pagenumbering{arabic}

%\thatchaphol{Things to fixed in the next version: (1) the success probability of $1-p$ should be in the statements of many theorems, because we use the algorithms recursively in small graphs. (2) Thm 2.2 now looks deterministic, (3) statement of expander pruning (what happen when $G_0$ is not expander)}

\section{Introduction}
Graph clustering algorithms are extensively studied and have wide practical
applications such as unsupervised learning, community detection, and image segmentation
(see e.g. \cite{FORTUNATO201075,schaeffer2007graph,ShiM00}). A natural bicriteria notion for graph clustering 
introduced by Kannan Vempala and Vetta \cite{KannanVV00}, which we refer to as \emph{expander
decomposition}, is to decompose a graph into \emph{clusters} such that each cluster is richly intraconnected and sparsely connected to the rest of the graph. More formally, given a graph $G=(V,E)$ we aim to find a partitioning of $V$ into $V_{1},\dots,V_{k}$ for some $k$, such that the total number of
edges across different clusters is small while the \emph{conductance} of
each cluster as an induced subgraph is large. This bicriteria measure is advantageous 
over other popular measures such as min diameter decomposition,
$k$-center, and $k$-median since there are simple examples
where these measures fail to capture the natural clustering. Moreover, expander decomposition has seen great applications in algorithm design
 including graph sketching/sparsification \cite{AndoniCKQWZ16,JambulapatiS18,ChuGPSSW18},
undirected/directed Laplacian solvers \cite{SpielmanT04,CohenKPPRSV17}, 
max flow algorithms \cite{KelnerLOS14}, approximation algorithms
for unique game \cite{Trevisan05}, and dynamic minimum spanning forest
algorithms \cite{NanongkaiS16,Wulff-Nilsen16a,NanongkaiSW17}. An
efficient algorithm to compute expander decomposition is crucial
for all these applications (See \ref{sub:app} for a more detailed discussion
about applications). With the abundance of massive graphs, it is crucial to design algorithms with running time 
at most nearly linear in the size of the graph, and thus nearly linear time expander decomposition methods are of great 
interest both in theory and in practice.

To continue the discussion, we need to introduce some notations. For an undirected graph $G=(V,E)$, we denote $\deg(v)$ as the number of edges incident to $v\in V$, and $\vol(C)=\sum_{v\in C}\deg(v)$ as the {\em volume} of $C\subseteq V$. We use subscripts to indicate what graph we are working with, while we omit the subscripts when the graph is clear from context. A {\em cut} is treated as a subset $S\subset V$, or a partition $(S, \overline{S})$ where $\overline{S} = V\setminus S$. For any subsets $S,T\subset V$, we denote $E(S,T)=\{\{u,v\}\in E\mid u\in S,v\in T\}$ as the set of edges between $S$ and $T$. The {\em cut-size} of a cut $S$ is $\delta(S)=|E(S,\overline{S})|$. The {\em conductance} of a cut $S$ in $G$ is $\Phi_G(S)= \frac{\delta(S)}{\min(\vol_G(S),\vol_G(V\setminus S))}$. Unless otherwise noted, when speaking of the
conductance of a cut $S$, we assume $S$ to be the side of smaller volume. The conductance of a graph $G$ is $\Phi_{G}=\min_{S\subset V}\Phi_G(S)$. 
If $G$ is a singleton, we define $\Phi_{G}=1$. Let $G[S]$ be the subgraph induced by $S\subset V$, and we denote $G\{S\}$ as the induced subgraph $G[S]$ but with self-loops\footnote{Each self-loop contributes $1$ to the degree of a node.} added to vertices so that any vertex in $S$ has the
same degree as its degree in $G$. Observe that for any $S\subset V$,
$\Phi_{G[S]}\ge\Phi_{G\{S\}}$. We say a graph $G$ is a $\phi$ \emph{expander} if $\Phi_{G}\ge\phi$, and we call a partition $V_{1},\dots,V_{k}$ of $V$ 
a $\phi$ expander decomposition if $\min_{i}\Phi_{G[V_{i}]}\ge\phi$.

We first note that for any graph, there always \emph{exists} an expander decomposition
with the following guarantee:
\begin{observation}
[Ideal expander decomposition]\label{fact:ideal}Given any graph $G=(V,E)$ with $m$ edges and a parameter $\phi\in(0,1)$, there
exists a partition $V_{1},\dots,V_{k}$ of $V$ for some $k$ such
that $\min_{i}\Phi_{G[V_{i}]}\ge\phi$ and $\sum_{i}\delta(V_{i})=O(\phi m\log n)$.
\end{observation}
The trade off in \Cref{fact:ideal} is tight. As observed in \cite{AlevALG18}, in a hypercube after deleting a small constant fraction of edges, some remaining components will have conductance $O(1/\log n)$. There is an simple and algorithmic argument for \ref{fact:ideal} as follows. Given a graph $G$, we find the cut $S$ with the smallest conductance (or interchangeably referred to as the sparsest cut in our discussion) . If $\Phi_G(S)<\phi$, we cut along $S$, and recursively decompose $G[S]$ and $G[\overline{S}]$. Otherwise, we know $\Phi_{G}\ge\phi$ so we can output $V$ as an expander cluster, and it is clear all the output clusters are $\phi$ expanders. The number of edges across different clusters follow from a simple charging argument, where each time we cut the graph, we charge the number of edges we cut to the edges remaining in the smaller side of the cut. Since the cuts have conductance at most $\phi$, and any edge can be on the smaller side of a cut at most $O(\log m)$ times, we know the total number of edges across clusters is at most $O(\phi m\log m)$.

Since finding the sparsest cut is NP-hard, a natural way to turn the above argument into an efficient algorithm is to use approximate sparsest cuts instead. This gives polynomial time algorithms at the cost of leaving more edges across the clusters. If we fix the conductance requirement $\phi$ of the expander decomposition, i.e. $\min_{i}\Phi_{G[V_{i}]}\ge\phi$, the quality of an expander decomposition algorithm is characterized by two measures: $\sum_{i}\delta(V_{i})$, which we refer to as the {\em error}, and the running time. We want a fast algorithm with error almost as small as in \ref{fact:ideal}. 

\subsection{Previous Work}
There are various efficient algorithms that can compute an approximate sparsest cut. These methods mostly fall into two categories. One type of the methods are spectral based, that is, they use the eigenvalue and eigenvectors of the graph Laplacian matrix or the random walk diffusion process. These methods can have quadratic error in the cut quality, which means that if the sparsest cut has conductance $\gamma$, the cut found by these methods can have conductance as large as $\Omega(\sqrt{\gamma})$. The quadratic loss is inherent to spectral methods as observed by Cheeger \cite{Cheeger69}. Another type of methods are based on single-commodity or multi-commodity flow techniques, and these flow-based methods can find cut with conductance that is at most a $\log^{O(1)} m$ factor worse than the optimal conductance (\cite{OrecchiaSVV08,AroraRV09,KhandekarRV09,Sherman09}). 

One can directly apply the approximate sparsest cut methods in the recursive decomposition approach, and get polynomial time expander decomposition algorithms of various error bounds (\cite{KannanVV00}). However, as the approximate sparsest cut found can be very unbalanced, i.e. one side of the cut has much smaller volume than the other side, the recursion
depth can be $\Omega(n)$, and thus this approach inherently takes $\Omega(n^{2})$ time which is too slow for many applications.

The issue here is that black-box usage of approximate sparsest cut methods may return a very unbalanced cut even when a balanced sparse cut exists. Indeed, efficiently certifying there is no balanced sparse cut when the algorithm finds an unbalanced sparse cut is the main challenge in this line of work. Previous results achieving below quadratic running time all utilize nearly linear time subroutines to find approximately most balanced sparse cuts. In the case where the subroutines find a very unbalanced sparse cut, they also provide certificates that no sparse cut of much better balance exists. These certificates can usually be interpreted as the larger side of the unbalanced cut having certain well-connected properties, but none of them is strong enough to certify the larger side induces an expander. We discuss two representative prior results in more detail in the following. Henceforth, we use $\widetilde{O}(\cdot)$ to hide a
$\mbox{polylog}(n)$ factor.

First, Spielman and Teng \cite{SpielmanT04} show that %in $\widetilde{O}(m/\mbox{poly}(\phi))$ time, 
one can find a low conductance cut $(S,\overline{S})$ with the following
guarantee: either $(S,\overline{S})$ is balanced (i.e. $\min\{\vol(S),\vol(\overline{S})\}=\Omega(m)$),
or the larger side $G[\overline{S}]$ is contained in \emph{some } unknown expander
subgraph. Their algorithm just settles with the weaker guarantee on $G[\overline{S}]$, and thus \emph{avoids} recursing on the larger side of the unbalanced cut. This makes the recursion depth to be $O(\log n)$, and their algorithm takes $\widetilde{O}(m/\mbox{poly}(\phi))$ time and has error $\sum_{i}\delta(V_{i})=\widetilde{O}(\sqrt{\phi}m)$. As 
$G[\overline{S}]$ may not induce an expander, their
decomposition only guarantees the \emph{existence} of some unknown set $W_{i}\supseteq V_{i}$ where $\Phi_{G[W_{i}]}\ge\phi$ for each cluster
$V_i$ in their output. They show that the weaker expansion guarantee is sufficient for the application of spectral graph sparsification \cite{SpielmanT11}, and their remarkable result has since become the building block in several breakthroughs of fast graph algorithms (e.g. \cite{KelnerLOS14,CohenKPPRSV17}). 
However, there are other applications that crucially need the stronger guarantee of
$\Phi_{G[V_{i}]}\ge\phi$ for each cluster, e.g. dynamic minimum spanning forest algorithms
\cite{NanongkaiS16,Wulff-Nilsen16a,NanongkaiSW17} and short-cycle
decomposition \cite{ChuGPSSW18}. 

The second result is by the independent work of Nanongkai and Saranurak
\cite{NanongkaiS16} and Wulff-Nilsen \cite{Wulff-Nilsen16a}. By
using an approximate balanced sparse cut algorithm \cite{KhandekarRV09,Peng14}
in a black-box manner, they get a low conductance cut $(S,\overline{S})$
such that, if the cut is not balanced, then any subset $T$ of the larger side $\overline{S}$ must have high conductance in $\overline{S}$ if $\vol(T)\ge k$ for some $k$.\footnote{More formally, we mean $|E(T,\overline{S}-T)|\ge\phi \vol_{G[\overline{S}]}(T)$ for all $T\subset S$ where $k\leq \vol_{G[\overline{S}]}(T)\leq \vol(G[\overline{S}])/2$.} Note that if $k=1$, then $G[\overline{S}]$
would have been an expander. They show how to iteratively reduce $k$
to $1$ by recursing on the larger side $m^{o(1)}$
many times until they obtain an expander subgraph. This approach gives real expander decomposition (i.e. $\forall i:\Phi_{G[V_{i}]}\ge\phi$), takes total running time $O(m^{1+O(\log\log n/\sqrt{\log n}}))$, and the final decomposition has error $\sum_{i}\delta(V_{i})=O(\phi m^{1+O(\log\log n/\sqrt{\log n}})$.

\subsection{Our contribution.}
In this paper, we simultaneously improve upon both
previous works. That is, similar to \cite{SpielmanT04}, we do not
recurse on the larger side of the unbalanced cut, while at
the same time every output cluster induces an expander like in \cite{NanongkaiS16,Wulff-Nilsen16a}. Apart from expander decomposition, our main technique also lead to better expander pruning for dynamic graphs. We discuss them separately.

\paragraph{Expander Decomposition} 

Our result is the first nearly linear time expander decomposition algorithm in the regime of $\phi$ being at least $1/\log^{O(1)} m$, which is the case for most applications of expander decomposition \cite{AndoniCKQWZ16,JambulapatiS18,ChuGPSSW18,SpielmanT04,CohenKPPRSV17,Trevisan05,NanongkaiSW17}.
%This is achieved by showing a $\widetilde{O}(m/\phi)$-time algorithm which
%finds a low conductance cut $(S,\overline{S})$ such that, if $(S,\overline{S})$
%is not balanced, then the larger side $G[\overline{S}]$ is an expander.
%A $\widetilde{O}(m/\mbox{poly}(\phi))$-time expander decomposition was noted as an open problem by Spielman and Teng \cite{SpielmanT04}\footnote{See this version of the paper: https://arxiv.org/abs/cs/0310051v9}, by Koutis and Miller \cite{KoutisM08}, and by Orecchia Sachdeva and Vishnoi \cite{OrecchiaV11,OrecchiaSV12}. We give an
%affirmative answer as follows:
\begin{thm}
[Expander Decomposition]\label{thm:expander decomp}Given a graph $G=(V,E)$ of $m$ edges and a parameter $\phi$, there is a randomized algorithm that with high probability finds a partitioning of $V$ into $V_{1},\dots,V_{k}$ such that $\forall i:\Phi_{G[V_{i}]}\ge\phi$
and $\sum_{i}\delta(V_{i})=O(\phi m\log^{3}m)$. In fact, the algorithm
has a stronger guarantee that $\forall i:\Phi_{G\{V_{i}\}}\ge\phi$. The running
time of the algorithm is $O(m\log^{4}m/\phi)$.
\end{thm}

Note our result also has the right dependence (up to $\mbox{polylog}n$ factor) on $\phi$ in the error bound according to \ref{fact:ideal}. The construction by Spielman and Teng \cite{SpielmanT04}, even after using the state-of-the-art spectral-based algorithms for finding sparse cuts \cite{AndersenCL06,AndersenP09,OrecchiaV11,GharanT12,OrecchiaSV12}, can only guarantee $\sum_{i}\delta(V_{i})=\widetilde{O}(\sqrt{\phi}m)$ due to the intrinsic Cheeger barrier of spectral methods. Beyond the theoretical improvements, we note that comparing to prior work, our result only relies on techniques that are fairly basic and simple, and thus is very likely to have practical significance.

Our result extends to weighted graphs in a fairly straightforward way as in \ref{thm:main-weighted}, see \ref{section:weighted} for details. For simplicity, we focus on the unweighted case in our presentation.

\paragraph{Expander Pruning}
%\thatchaphol{Motivate more. Say explicitly that this is the key for breaking the 20-year-old bound of dynamic MST from $O(\sqrt{n})$ to $n^{o(1)}$ worst-case. Say that this is NOT extended to weighted graphs. Say briefly why previous result are slow and bad quality (they don't reuse flow).}

Although the utility of expander decomposition has been well-known for static problems, its applications in dynamic problems has only been explored fairly recently. Nanongkai, Saranurak, and Wulff-Nilsen \cite{NanongkaiSW17} significantly improved the 20-year-old $O(\sqrt{n})$ worst-case update time \cite{Frederickson85,EppsteinGIN97} of dynamic minimum spanning forest to only $n^{o(1)}$ time by using the expander decomposition.
However, as expander decomposition is a static object, they need a key tool in addition, which is referred to as \emph{expander pruning}.
Expander pruning is an algorithm for maintaining an expander under edge deletions, and we give the following result.
%Using our local-flow based technique, we obtain an significantly improved expander pruning algorithm as follow:
%A key tool called \emph{expander pruning} enables several dynamic algorithms \cite{NanongkaiS16,Wulff-Nilsen16a,NanongkaiSW17} to exploit
%expander decompositions, as it maintains
%each expander cluster under edge deletions as follows:
\begin{thm}[Expander Pruning]
\label{thm:pruning}Let $G=(V,E)$ be a $\phi$ expander with $m$ edges. There is a deterministic
algorithm with access to adjacency lists of $G$ such that, given
an online sequence of $k\le\phi m/10$ edge deletions in $G$, can
maintain a \emph{pruned set} $P\subseteq V$ such that the following
property holds. Let $G_i$ and $P_{i}$ be the graph $G$ and the set $P$ after the $i$-th deletion.
We have, for all $i$, 
\begin{enumerate}
\item $P_{0}=\emptyset$ and $P_{i}\subseteq P_{i+1}$,
\item $\vol(P_{i})\le8i/\phi$ and $|E(P_{i},V-P_{i})|\le4i$, and
\item $G_i\{V-P_{i}\}$ is a $\phi/6$ expander. 
\end{enumerate}
The total time for updating $P_{0},\dots,P_{k}$ is $O(k\log m/\phi^{2})$. 
\end{thm}
This significantly improves in
many ways the best known result
by Nanongkai, Saranurak, and Wulff-Nilsen \cite{NanongkaiSW17}.%, which in turns improve the previous algorithms in \cite{NanongkaiS16,Wulff-Nilsen16a}. 
In Theorem 5.2 of \cite{NanongkaiSW17}; 1) all the deletions
must be given in one batch and the algorithm only outputs the set $P$
for the graph after all edges in the batch are deleted, 2) their slower running
time is $\tilde{O(}\frac{\Delta k^{1+\delta}}{\phi^{6+\delta}})$
where $\Delta$ is the max degree of $G$ and $\delta\in(0,1)$, and
3) they only guarantee that $\Phi_{G[V-P]}=\Omega(\phi^{2/\delta})$
which is much lower than $\phi/6$. 

Their result is obtained by calling a local-flow subroutine \cite{OrecchiaZ14,HenzingerRW17} in a black-box way for many rounds without reusing flow information from previous rounds.
On the other hand, by using the \emph{trimming} technique in \ref{section:trimming} for constructing the expander decomposition in \Cref{thm:expander decomp} which can ``reuse'' the local flow across many rounds, \ref{thm:pruning} is obtained almost immediately.
%See its applications in \ref{sub:app}.

\ref{thm:pruning} is one of very few dynamic algorithms whose amortized
update time of $O(\log m/\phi^{2})$ is guaranteed over a short sequence
of updates (i.e. $k\le\phi m/10$). Previous results in
the literature only guarantees that given a sequence of $L$ updates to a graph with $m$ initial edges, the total update time is $O((m+L)T)$ for some $T$. This only gives $O(T)$ amortized update time when $L$ is $\Omega(m)$ (see e.g. \cite{EvenS81,HenzingerK99,HolmLT01}). 

Again the result extends to weighted graphs in a fairly straightforward way. For long sequence of updates, we can show the first expander pruning on weighted graphs with small amortized update time (\ref{thm:pruning-weighted}). We expect this would enable future dynamic algorithms to exploit the expander decomposition on weighted graphs. See \ref{section:weighted} for details.

\section{Overview}
In this section we discuss the high-level ideas of our algorithm, and show how these lead to our main result \ref{thm:expander decomp}.
Our algorithm (\ref{alg:exp-decomp}) use the recursive decomposition framework, where we try to find a sparse cut that is as balanced as possible, and recurse on both sides if the cut is balanced up to some $\mbox{polylog}n$ factor. However, if the sparse cut $(A,\overline{A})$ we find is very unbalanced, the key observation is that as long as we have certain weak expansion guarantee (\ref{defn:nearlyexp}) on the larger side $G[A]$, we can efficiently find another sparse cut $(A',\overline{A'})$ such that we can certify the larger side $G[A']$ is an expander. Our key technical contribution is a fast and simple method to perform some local fixing in $G[A]$ around the cut $(A,\overline{A})$ to obtain such $A'$. As a consequence, we only need to recurse on the smaller side $G[\overline{A'}]$ in this case, and the overall recursion has at most $O(\mbox{polylog}n)$ depth, while all clusters we find are induced expanders. We refer to this method of finding $A'$ as the \emph{trimming} step, and formalize its performance as follows.
\begin{thm}
[Trimming]\label{thm:trimming}
Given graph $G=(V,E)$ and $A\subset V, \overline{A}=V\setminus A$ such that $A$ is a nearly $\phi$ expander in $G$, $|E(A,\overline{A})|\leq \phi\vol(A)/10$, the trimming step finds $A' \subseteq A$ in time $O\left(\frac{|E(A,\overline{A})|\log m}{\phi^2}\right)$ such that $\Phi_{G\{A'\}}\geq \phi/6$. Moreover, $\vol(A')\geq \vol(A)-4|E(A,\overline{A})|/\phi$, and $|E(A',\overline{A'})|\leq 2|E(A,\overline{A})|$.
\end{thm}
The notion of a nearly $\phi$ expander, which is formally defined in \ref{defn:nearlyexp}, is the weak expansion guarantee we require for the larger side $G[A]$ of the unbalanced cut. We can usually get this guarantee from an approximate balanced sparse cut subroutine when it returns a very unbalanced cut. For example, the guarantee in the result of Spielman and Teng \cite{SpielmanT13} that the larger side of the unbalanced cut is inside some unknown larger expander immediately implies the larger side is a nearly expander. Other methods for balanced sparse cut (e.g. \cite{OrecchiaV11,OrecchiaSV12}) also provide certificates in the case of a very unbalanced cut that typically can imply the nearly expander guarantee on the larger side fairly straightforwardly. The upper-bound on $|E(A,\overline{A})|$ as a condition in \ref{thm:trimming} holds since the cut $(A,\overline{A})$ has low conductance and $\vol(\overline{A})$ is much smaller than $\vol(A)$ (i.e. unbalanced). 

Our trimming step relies on recent development on local flow algorithms \cite{OrecchiaZ14,HenzingerRW17}, and adapts (in a white-box manner) the particular push-relabel based local flow algorithm from \cite{HenzingerRW17}. In particular, we show that our problems boils down to certifying the non-existence of certain local bottleneck structures in the induced subgraph $G[A']$ where $A'$ is some subgraph of $A$, and we use the flow solution of a carefully designed flow problem as the certificate. In our trimming step, we construct such a flow solution in rounds, and there are possibly many rounds. The main challenge is that both the graph and the flow problem evolve across the rounds, so the running time can be very slow if we apply previous local flow methods (or nearly linear time approximate%\thatchaphol{We should say this is ``approximate'' in what sense. It is a bit confusing for me.} 
 max flow) as black-box in each round, as the number of rounds can be large. Instead, we adapt our flow subroutine in a way to make it \emph{dynamic}, so that if the graph and flow demands only change a little bit across two rounds, we can quickly update the flow solution from the previous round to get a new flow solution instead of computing from scratch. This allows us to bound the total running time as long as we can bound the total amount of change on the flow problems across all rounds rather than the total number of rounds. This is also the key insight of how our techniques give improved algorithms on dynamic graphs.

Given the trimming step, one can pair it with various balanced sparse cut methods to develop expander decomposition algorithms. To keep our discussion concrete and complete, we focus on a specific method which is a fairly standard adaptation of the {\em cut-matching framework} by Khandekar, Rao and Vazirani \cite{KhandekarRV09}. The choice of this particular method is due to its simplicity and robustness. Moreover, the basic flow subroutine we use in the trimming step can be easily integrated to the cut-matching framework, so we can avoid black-box usage of a hammer such as approximate max flow, which has a horrendous $\mbox{polylog}n$ factor in the running time, and is also not easy to implement or even capture with simpler heuristics. We refer to our method as the \emph{cut-matching} step, with the following guarantee.
\begin{thm}[Cut-Matching]\label{thm:cut-matching}
Given a graph $G=(V,E)$ of $m$ edges and a parameter $\phi$, the cut-matching step takes $O((m\log m)/\phi)$ time and must end with one of the three cases:
\begin{enumerate}
\item We certify $G$ has conductance $\Phi_G\geq \phi$.
\item We find a cut $(A,\overline{A})$ in $G$ of conductance $\Phi_G(\overline{A})=O(\phi\log^2 m)$, and $\vol(A),\vol(\overline{A})$ are both $\Omega(m/\log^2 m)$, i.e., we find a relatively balanced low conductance cut.
\item We find a cut $(A,\overline{A})$ with $\Phi_G(\overline{A})\le c_0\phi\log^2 m$ for some constant $c_0$, and $\vol(\overline{A})\leq m/(10c_0\log^2 m)$, and $A$ is a nearly $\phi$ expander (See \ref{defn:nearlyexp}).
\end{enumerate}
\end{thm}
\begin{algorithm}
	\caption{Expander Decomposition}
	\label{alg:exp-decomp}
	\fbox{
		\parbox{0.97\columnwidth}{
			Decomp$(G,\phi)$ \\ 
			\tab {\bf Call} Cut-Matching($G$,$\phi$)\\
			\tab {\bf If} we certify $\Phi_G\geq \phi$, {\bf Return} $G$ \\
			\tab {\bf Else if} we find relatively balanced cut $(A,R)$\\
			\tab \tab {\bf Return} Decomp$(G\{A\},\phi)$ $\bigcup$ Decomp$(G\{R\},\phi)$\\
			\tab {\bf Else} (i.e., we find very unbalanced cut $(A,R)$)\\
			\tab \tab $A'$=Trimming$(G,A,\phi)$\\
			\tab \tab {\bf Return} $A'$ $\bigcup$ Decomp$(G\{V\setminus A'\},\phi)$.
	}}
\end{algorithm}
%This step is based on the {\em cut-matching framework} by Khandekar, Rao and Vazirani \cite{KhandekarRV09}, which is designed to either certify a graph having high conductance or find a low conductance cut. If the cut is balanced, then we get the second case. However, the main issue from using this framework in a ``black-box'' manner is that we may find a very unbalanced cut and cannot afford to just recurse on both sides.
 %
%It turns out that we can adjust the framework so that when we find a very unbalanced cut, we can remove the smaller side, and continue the effort to eventually certify that the larger side is a nearly expander. 
As our cut-matching step is a fairly straightforward adaptation of the work of Khandekar et al. \cite{KhandekarRV09}, and is similar to how R\"{a}cke, Shah and T\"{a}ubig \cite{RackeST14} adjust the cut-matching framework in the context of oblivious routing, we defer further description and analysis of it to \Cref{section:cut-matching}. We note that the setting we apply our cut-matching step is more regularized than the one in \cite{RackeST14}, so our adaptation and analysis are considerably simpler and more basic comparing to \cite{RackeST14}.

Given the cut-matching step and the trimming step, we can combine the guarantee from case $(3)$ of the cut-matching step (\ref{thm:cut-matching} with the trimming step (\ref{thm:trimming}) to write out explicitly the quality of the cut from the trimming step. We have $\vol(\overline{A})\leq m/(10c_0\log^2 m)$ and $\Phi_G(\overline{A})\leq c_0\phi\log^2 m$ from case$(3)$ of the cut-matching step. So $|E(A,\overline{A})|\leq \phi m/10$. Thus, the trimming step in $O\left((m\log m)/\phi\right)$ time gives a $\phi/6$ expander $G\{A'\}$ where $\vol(A')\geq 3m/2$ (note the total volume is $2m$). \ref{thm:trimming} also indicates that the conductance of $(A',V\setminus A')$ is at most twice the conductance of $(A,V\setminus A)$, so $(A',V\setminus A')$ has conductance $O(\phi\log^2 m)$. Now we can give a quick proof of our main result.

\noindent
\fbox{
	\parbox{0.97\columnwidth}{
{\bf Proof of \ref{thm:expander decomp}.}
Since our recursive algorithm only stops working on a component when it certifies the induced subgraph has conductance at least $\phi/6$, the leaves of our recursion tree give an expander decomposition $V_1,\ldots,V_k$ of $V$, and $\Phi_{G\{V_i\}}\geq \phi/6$ for all $i\in[1,k]$. Note that whenever we cut a component and recurse, we always add self-loops so the degree of a node remains the same as its degree in the original graph, so we get the stronger expansion guarantee with respect to the volume in the original graph. 

We now bound the running time. After carrying out the cut-matching step on a component, if we get case $(1)$ in \ref{thm:cut-matching}, we are done with the component. In case $(2)$, we get a relatively balanced sparse cut, and we recurse on both sides of the cut. In case $(3)$, we use the trimming step to get a sparse cut such that the larger side of the cut is a $\phi/6$ expander, and we only recurse on the smaller side of the cut. In any case, we get the volume of the largest component drops by a factor of at least $1-\Omega(1/\log^2 m)$ across each level of the recursion, so the recursion goes up to $O(\log^3 m)$ levels. As the components on one level of the recursion are all disjoint, the total running time on all the components of one level of the recursion is $O\left((m\log m)/\phi\right)$, so the total running time is $O\left((m\log^4 m)/\phi\right)$. 

To bound the number of edges between expander clusters, observe that in both case $(2)$ and case $(3)$, we always cut a component along a cut of conductance $O(\phi\log^2 m)$. Thus, we can charge the edges on the cut to the edges in the smaller side of the cut, so each edge is charged $O(\phi\log^2 m)$. 
%If we cut in case $(3)$ (i.e., trimming step), we charge the cut edges to the edges in $A'$, and each edge is charged $O(\phi)$, and will only be charged once this way.
%If we cut in case $(3)$, observe that the cut $(A',V-A')$ obtained from the trimming step has conductance $O(\phi\log^2 m)$. So, again, each edge on the smaller side $V-A'$ is charged $O(\phi\log^2 m)$.
 An edge can be on the smaller side of a cut at most $\log m$ times, so we can charge each edge at most $O(\phi\log^3 m)$ to pay for all the edges we leave between the final clusters. This bound the total number of edges between the expanders to be at most $O(m\phi\log^3 m)$. }}

In the rest of this extended abstract, we discuss the trimming step in \ref{section:trimming} with some details of the flow subroutine deferred to \ref{app:unit-flow}. We briefly discuss the (theoretical) applications of our result and open problems in \ref{section:discuss}. For completeness, we discuss the cut-matching step in \ref{section:cut-matching}. We will keep our discussion at a high level, and leave most of the technical details to the full version.

\section{The Trimming Step}
\label{section:trimming}
%\paragraph{Nearly expander and feasible routing.}
In this section, we describe the trimming step and prove \Cref{thm:trimming}, which is the key technical contribution of this paper. We start with defining a nearly expander formally and introducing the flow terminology we use in our subroutine.
%Recall when we take the trimming step, we have a set $A$ that is a nearly $\phi$ expander, which we formally define as follows.
\subsection{Preliminaries}
\label{sec:flow_prelim}
\begin{definition}[Nearly Expander]
\label{defn:nearlyexp}
Given $G=(V,E)$ and a set of nodes $A\subset V$, we say $A$ is a {\em nearly $\phi$ expander} in $G$ if 
\[
\forall S\subseteq A, \vol(S)\leq \vol(A)/2 : |E(S,V\setminus S)| \geq \phi\vol(S)
\]
Note if the left hand side of the inequality is $|E(S,A\setminus S)|$, $G\{A\}$ would be a $\phi$ expander.    
\end{definition}
The trimming step aims to find a $A'\subseteq A$ that $G\{A'\}$ is a $\phi/6$ expander, i.e.,
\[
\forall S\subseteq A', \vol(S) \leq \vol(A')/2 : |E(S,A'\setminus S)| \geq \phi\vol(S)/6
\]  
Recall when we consider induced subgraphs, we always add self-loops to nodes so that the degree of a node is the same as its degree in the original graph,
%(where each self-loop contributes one to the degree)
 so there is no ambiguity on the notation $\vol(S)$. 

 A {\em flow problem} $\Pi$ on  a graph $G=(V,E)$ is specified by a source function $\Delta: V\rightarrow \mathbb{R}_{\geq 0}$, a sink function $T:V\rightarrow \mathbb{R}_{\geq 0}$, and edge capacities $c:E\rightarrow \mathbb{R}_{\geq 0}$.
We use {\em mass} to refer to the substance being routed.
For a node $v$, $\Delta(v)$ specifies the amount of mass initially placed on $v$, and $T(v)$ specifies the capacity of $v$ as a sink. For an edge $e$, $c(e)$ bounds how much mass can be routed along the edge.

A {\em routing} (or {\em flow}) $f:V\times V \rightarrow \mathbb{R}$ satisfies $f(u,v)=-f(v,u)$ and  $f(u,v)=0$ for $\{u,v\}\notin E$. $f(u,v)>0$ means that mass is routed in the direction from $u$ to $v$, and vice versa. If $f(u,v) = c(e)$, then we say $(u,v)$ is {\em saturated} (in the direction from $u$ to $v$).  
Given $\Delta$, we also treat $f$ as a function on vertices, where $f(v)=\Delta(v)+\sum_u f(u,v)$ is the amount of mass ending at $v$ after the routing $f$. 
If $f(v) \ge T(v)$, then we say $v$'s sink is {\em saturated}.

$f$ is a {\em feasible routing/flow} for $\Pi$ if $|f(u,v)| \leq c(u,v)$ for each edge $e = \{u,v\}$ (i.e. obey edge capacities), $\sum_u f(v,u) \leq \Delta(v)$ for each $v$ (i.e. the net amount of mass routed away from a node can be at most the amount of its initial mass), and $f(v) \leq T(v)$ for each $v$. These notations are more natural in the discussion of local flow methods. 
% For any sets $S,S' \subset V$, let $\Delta(S) =\sum_{v\in S} \Delta(v)$ and  $T(S)= \sum_{v\in S} T(v)$ and $c(S,S')= \sum_{e\in E(S,S')} c(e)$. Observe that if there exists $S$ where $\Delta(S) > T(S) + c(S,V-S)$, then there cannot exists a feasible routing.

\subsection{Finding an Expander inside a Nearly Expander}
In this section, we describe a generic algorithm for finding an expander inside a nearly expander.
%
%Now we discuss the intuition of why the trimming step is possible, and give a clean way to perform it using exact max flow computation which takes super-linear time (i.e. slow trimming).
The algorithm is based on the following simple observation: If $A$ is a nearly expander in $G$, but $G\{A\}$ is not an induced expander, then any low conductance cut in $G\{A\}$ must be ``close'' to the places where $A$ is cut off from $G$. More formally, consider the following flow problem in $G\{A\}$. We let each edge in $E(A,V\setminus A)$ be a source of $2/\phi$ units of mass\footnote{
	We slightly abuse the notation to place initial mass on an edge that no longer exists in $G\{A\}$. As a cut edge $\{u,v\}\in E(A,V\setminus A)$ with $v\in A$ corresponds to a self-loop attached to $v$ in $G\{A\}$, technically $v$ is where we actually place the initial mass. Note if $x$ edges in $E(A,V\setminus A)$ are incident to $v$, the amount of source mass on $v$ will be $2x/\phi$ units. 
%	In our discussion we stick with the notation of treating each cut edge as a source of $2/\phi$ units of mass to avoid the complication of computing the actual amount of source mass on nodes.
}, 
each node $v$ be a sink of capacity equal to its degree $\deg(v)$%\footnote{By a sink of capacity $x$ we mean we are allowed to route at most $x$ units of mass to the node}
, and edges all have capacity $2/\phi$. See \Cref{fig:flowproblem} for an example. We have the following:
\begin{figure}
	\begin{centering}
		\includegraphics[viewport=90bp 150bp 550bp 460bp,clip,scale=0.6]{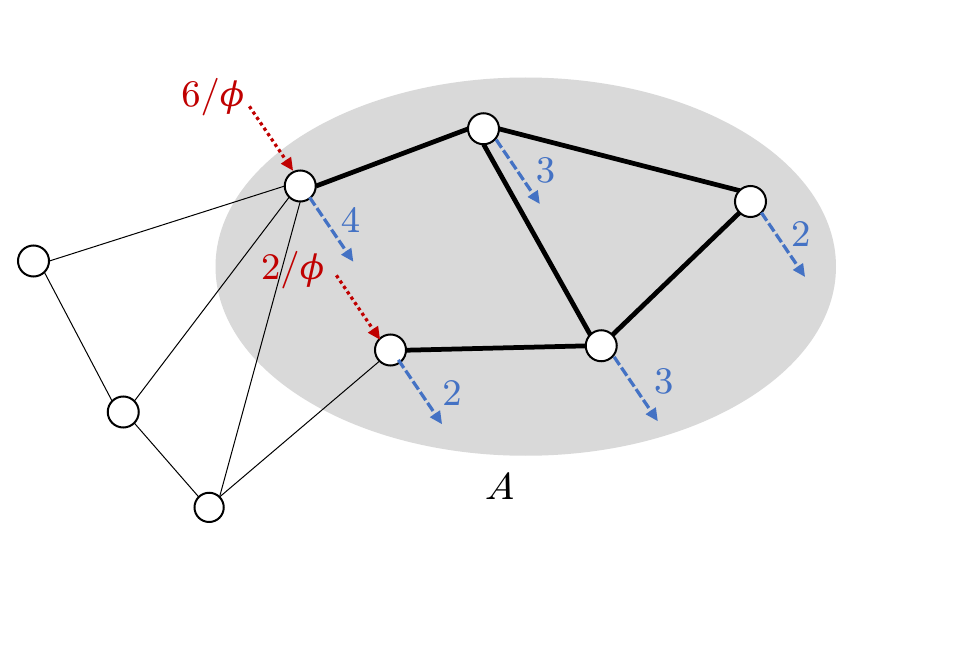}
		\par\end{centering}
	\caption{The flow problem for \ref{prop:certify}. The blue-dashed arrows and corresponding
		numbers indicate the sink capacity of edge nodes. The red-dotted arrows
		and corresponding numbers indicate the source function. Each black-bold
		edge is an edge in $G\{A\}$ and has edge capacity of $2/\phi$. Each black-thin edge is an
		edge outside $G\{A\}$. \label{fig:flowproblem}}
	
\end{figure}
\begin{prop}\label{prop:certify}
	Suppose that $A\subset V$ is a nearly $\phi$ expander in $G$, but $G\{A\}$ is not a $\phi/6$ expander. Then, the flow problem constructed as above doesn't have a feasible routing.
% \footnote{By a feasible routing in $G\{A\}$ we mean to route all the $2|E(A,V-A)|/\phi$ units of mass from sources to nodes in $A$, such that the amount of mass each node receives is at most its sink capacity, and the amount of mass routed along any edge is below the edge capacity. We formally introduce our notation for flow problems in \ref{sxn:flow}}.
\end{prop}
\begin{proof}
	Given the assumptions, there must be a $S\subset A, \vol(S) \leq \vol(A)/2$ such that
	
	 $$\frac{|E(S,A-S)|}{\vol(S)} \leq \phi/6,$$ but 
		$$\frac{|E(S,V-S)|}{\vol(S)}=\frac{|E(S,A-S)|+|E(S,V-A)|}{\vol(S)} \geq \phi.$$
	Let $a=|E(S,V-A)|,b=|E(S,A-S)|$, the above two inequalities give $a\geq 5b$ and $\vol(S)\leq \frac{6a}{5\phi}$. That is, any sparse cut $S$ in $G\{A\}$ must be local to where $A$ is cut off from the rest of the graph.
%	 Thus, if $G\{A\}$ is not a $\phi/6$ expander, any low conductance cut $S$ must have most of its boundary edges going out to $V-A$ compared to $A-S$. 
	
  Each of the $a=|E(S,V-A)|$ edges is a source of $2/\phi$ units of mass, so the total amount of source mass started in $S$ is $2a/\phi$. However, the total sink capacity of nodes in $S$ is $\vol(S)\leq 6a/(5\phi)$, and the amount of mass that can be routed out of $S$ to $A-S$ along the $b=|E(S,A-S)|$ edges is at most $2b/\phi\leq 2a/(5\phi)$ due to edge capacity. As $2a/\phi> 6a/(5\phi)+2a/(5\phi)$, there is too much initial mass in $S$ that we cannot route all the mass to sinks in $S$ or out of $S$ under edge and sink capacities and thus there can be no feasible routing of the flow problem we construct.  \end{proof}

Flow problems where each node is a sink with capacity proportional to its degree have been used previously in the literature (e.g. for finding densest subgraphs \cite{Goldberg84} and for improving cut quality \cite{LangR04,AndersenL08,OrecchiaZ14,VeldtGM16,WangFHMR17}). Note that \ref{prop:certify} holds if we use sink capacity $T(v)=c\deg(v)$ for any $c\leq 1$. The scalar $c$ dictates how local the flow computation needs to be, as long as $c$ is not too small so there is no feasible routing due to the trivial reason that the total sink capacity over the entire graph is not enough for the source mass. In our setting, the nearly expander guarantee allows us to use a fairly large $c=1$ and still recognize any sparse cut in $G\{A\}$ as long as there exists one. Thus, our flow computation can be very efficient, since it only needs to explore region very local to where we put initial mass. This local feature of our flow problem is why the running time of our trimming step is proportional to $|E(A,V\setminus A)|$ instead of $|E(G\{A\})|$, and enable us to extend the result to the dynamic setting (see \Cref{thm:pruning}). We note this idea has be exploited before to design local algorithms, for example, the local improve algorithm by Orecchia and Zhu \cite{OrecchiaZ14} comparing to the global improve algorithm by Andersen and Lang \cite{AndersenL08}. Intuition similar to our key observation \Cref{prop:certify} is also exploited by \cite{NanongkaiS16,NanongkaiSW17} for bounded degree graphs.

\begin{algorithm}
\caption{Trimming}
\label{alg:trimming}
\fbox{
\parbox{0.97\columnwidth}{
{\em Trimming}($G$,$A$,$\phi$) \\ 
\tab Set $A'=A$\\
\tab {\bf While} $G\{A'\}$ is not a certified $\phi/6$ expander\\
\tab \tab Construct the flow problem in $G\{A'\}$ where\\
\tab \tab \tab Each edge in $E(A',V-A')$ is a source of $2/\phi$ units of mass,\\
\tab \tab \tab Each node $v\in A'$ is a sink of capacity $\deg(v)$,\\
\tab \tab \tab Each edge in $G\{A'\}$ has capacity $2/\phi$.\\
\tab \tab Use a flow algorithm to find a feasible routing.\\
\tab \tab {\bf If} a feasible routing is found\\
\tab \tab \tab $G\{A'\}$ is a certified $\phi/6$ expander.\\
\tab \tab {\bf Else} (i.e. a cut is found in $G\{A'\}$ with $S$ being the small side)\\
\tab \tab \tab $A'=A'-S$
}}
\end{algorithm}

Given \Cref{prop:certify}, we can take a generic approach for our trimming step as in \Cref{alg:trimming}. We proceed in rounds, where we start with $A'=A$ in the first round, construct our flow problem in $G\{A'\}$, and use some flow algorithm to route the initial mass to sinks. If a feasible routing is found, as any subset $A'$ of $A$ is inherently a nearly $\phi$ expander in $G$, \Cref{prop:certify} certifies that $G\{A'\}$ is a $\phi/6$ expander. If the flow algorithm doesn't find a feasible routing, we will get a cut $S$ in $G\{A'\}$. In this case we trim $S$ (i.e. remove nodes in $S$ and their incident edges) from $A'$, and proceed to the next round. We do this iteratively until in some round our flow algorithm finds a feasible routing for the flow problem defined on $G\{A'\}$ to certify the $G\{A'\}$ in that round is a $\phi/6$ expander.

\subsection{Slow Trimming}
\label{sec:trimming_slow}

%Note in general we won't use an exact max flow algorithm for efficiency purpose, this case can happen even when a feasible routing exists, and the returned cut $S$ can be an approximate min cut instead of the min cut. 
As a warm-up, we show that if we solve the flow problem in \Cref{alg:trimming} using an exact max flow algorithm, then the algorithm will finish in at most $2$ rounds. Note that using exact max flow is too slow for us because currently all known algorithms take $m^{1+\Omega(1)}$ time (\cite{LeeS14,Madry16}). Nonetheless, this gives some intuition for our efficient algorithm in the next section. 

To be instructive, we will temporarily switch to an $s-t$ max-flow formulation, and translate the standard max-flow min-cut property to our flow language. 

\noindent
\fbox{
\parbox{0.97\columnwidth}{
{\bf Our flow problem in max-flow notation: } The flow problem in our discussion is equivalent to a $s-t$ max-flow problem in the augmented network of $G\{A\}$ where we add a super-source $s$, a super-sink $t$ to $G\{A\}$, add a directed edge of capacity $2/\phi$ from $s$ to each source in $E(A,V\setminus A)$, add a directed edge of capacity $\deg(v)$ from each $v\in A$ to $t$, and each original (undirected) edge in $G\{A\}$ has capacity $2/\phi$. 

\vskip 0.1in
Routing $x$ units of mass out of a source $e$ is equivalent to sending $x$ flow along the edge from $s$ to $e$, and routing $x$ units of mass to the sink of a node $v$ is equivalent to sending $x$ flow along the edge from $v$ to $t$. Routing all the mass from sources to sinks (i.e. a feasible routing in our notation) is equivalent to the $s-t$ max-flow problem having max-flow value $\frac{2|E(A,V-A)|}{\phi}$.
}}

If we are not done after the first round of trimming, we will have a max-flow $f$, and a corresponding min-cut $S$ in the augmented network of $G\{A\}$, where $S$ is the side containing the super-source $s$. 
The important fact about exact max-flow and min-cut is that $f$ will saturate all edges going out of $S$. This implies the following:
\begin{claim}
\label{claim:maxflow}
The standard max-flow min-cut properties translate $f$ and $S$ to a routing and cut in our notation with the following properties.
\begin{enumerate}
\item We can route $2/\phi$ units of mass from every source in $A\setminus S$ to sinks, because in the augmented network the edge from $s$ to any such source crosses the min-cut, so $f$ sends $2/\phi$ flow along the edge. Moreover, as $f$ only sends flow out of $S$, the mass from any source in $A\setminus S$ must be routed to sinks in $A\setminus S$ using only edges in $G\{A\setminus S\}$.
\item We can route $2/\phi$ units of mass from every edge in $E(S,A\setminus S)$ to sinks in $A\setminus S$ using only edges in $G\{A\setminus S\}$. Again this is due to $f$ saturating every edge across the cut from $S$ to $A\setminus S$.
\item Any node $v\in A$ that falls in $S$ must have its sink saturated, i.e. it receives $\deg(v)$ units of mass. This is because the edge $(v,t)$ in the augmented network crosses the min-cut.
\end{enumerate}
\end{claim}
In the second round, we have $A'=A\setminus S$. Consider the flow problem we construct on $G\{A'\}$ in the second round: the sources are the edges in $E(A',V\setminus A')$, which are of one of the two types
\begin{enumerate}
\item Sources in $E(A',V\setminus A')\bigcap E(A,V\setminus A)$: These are also sources in the first round that falls in $A\setminus S$. By \ref{claim:maxflow}$(1)$, we can route $2/\phi$ units of mass from each of of these sources to sinks in $G\{A'\}$.
\item Sources in $E(A',V\setminus A')-E(A,V\setminus A)$: These new sources are exactly the cut edges in $E(S,A\setminus S)$. By \ref{claim:maxflow}$(2)$, we can route $2/\phi$ units of mass from each of of these sources to sinks in $G\{A'\}$.
\end{enumerate}
Thus, the flow routing from the first round already gives us a feasible routing in the second round, and we can certify $G\{A'\}$ is a $\phi/6$ expander. Moreover, by \ref{claim:maxflow}$(3)$ every node $v$ in $S=A\setminus A'$ receives $\deg(v)$ units of mass in the routing of first round, so the total volume of the nodes we remove from $A$ is bounded by the total amount of mass $2|E(A,V-A)|/\phi$, which gives $\vol(A')\geq \vol(A)-2|E(A,V-A)|/\phi$. Moreover, as the routing in the first round routes $2/\phi$ units of mass out of every edge in $E(A',V\setminus A')$, and these mass are distinct since $f$ only routes mass from $S$ to $A\setminus S$, so $2|E(A',V\setminus A')|/\phi$ is also bounded by the total amount of mass $2|E(A,V-A)|/\phi$, which gives $|E(A',V-A')|\leq |E(A,V-A)|$.

The above discussion already gives (even a stronger version of) everything we want in \ref{thm:trimming} except the running time. 

%
%Following the above argument, our trimming step runs in rounds. We start with $A'=A$, and iteratively trim (i.e. remove) nodes and their incident edges from $A'$ until we can certify $G\{A'\}$ is a $\phi/6$ expander. Each round we construct the aforementioned flow problem in $G\{A'\}$: every cut edge in $E(A',V-A')$ is a source of $2/\phi$ units of mass, each node has sink capacity equal to its degree, and edges have capacity $2/\phi$. If we find a feasible flow routing in $A'$ we can certify $G\{A'\}$ is a $\phi/6$ expander. The reason is that if $A$ is a nearly $\phi$ expander in $G$, then from the definition it is easy to see that any $A'\subset A$ is also a nearly $\phi$ expander. Then by our previous argument, a feasible flow routing certifies there is no low conductance cut in $G\{A'\}$. If we cannot find a feasible routing, the flow algorithm we use will return a cut $S$ in $G\{A'\}$. We trim $S$ from $A'$ and proceed to the next round. The procedure keeps going until we finish with some $\phi/6$ expander $G\{A'\}$. 
%

\subsection{Efficient Trimming}
As we cannot afford to use exact max flow for efficiency purpose, in general we can only expect to find a pair of flow and cut that are approximately tight, that is, the flow saturates most of the outgoing capacity of the cut. As a result, there is no guarantee that if we remove a cut, the flow problem in the next round will have a feasible routing as in the exact max flow case. The main challenge is that the trimming step can take many rounds to converge, and although approximate max flow only takes nearly linear time to compute, it will be too slow for us to solve the flow problem in each round from scratch. This is the main reason it is not sufficient to directly apply efficient approximate max flow methods such as \cite{KelnerLOS14,Sherman13,OrecchiaZ14,HenzingerRW17}.

Consider the sequence of flow problems in our trimming step as described in \Cref{alg:trimming}, we can view it as a dynamic flow problem, where in each round we update the graph by trimming the cut found in the last round, and add new source mass to the graph. Thus, we need a flow subroutine to handle these flow problems dynamically instead of starting from scratch in each round. Note the changes to the flow problems are quite regularized, where the graph strictly shrinks across rounds, any source from one round either remains a source in the next round or is removed from the graph, and new sources are added exactly at the new cut edges. In the extreme case of exact max flow min cut, the argument is simple because of the strong optimality condition where the max flow $f$ saturates every edge from $S$ to $A\setminus S$. This allows us to reuse the flow in a trivial manner by constraining the routing $f$ to the subgraph $G\{A\setminus S\}$. If we only compute an approximately tight pair of flow and cut (instead of exact max-flow and min-cut), we get weaker optimality conditions such as the routing saturates most of the outgoing capacity of the cut (instead of all of the outgoing capacity). When we proceed from round $i$ to round $i+1$ with the cut removed, the routing still tells us how to route \emph{most} of the source mass in round $i+1$. This strongly suggests we should be able to reuse the flow routing across rounds, and only pay update time proportional to how much the flow problems have changed. 
\subsubsection{Reusing Flow Information}
\label{section:trimming_fast}
%To motivate our algorithm, let's first consider the ideal case where
%we solve a max-flow problem exactly. Here, \Cref{alg:trimming} will finish in only two
%rounds! If we do not get a feasible flow from the first round, the exact algorithm will return a minimum cut $S$ together with a routing $f$ with the following crucial property: every edge $(u,v)$
%where $u\in S$ and $v\in V-S$, $(u,v)$ is saturated in the direction
%from $u$ to $v$, i.e. $f(u,v)=2/\phi$. Once we consider $A'=A-S$,
%the routing $f$ immediately gives us the feasible flow for the flow
%problem on $G\{A'\}$ and we are done by reusing such routing. 
%We prove this fact formally in \Cref{sec:trimming_slow} because it can be instructive but is not necessary for understanding our final trimming algorithm.
%
%Although solving max flow exactly takes $m^{1+\Omega(1)}$ time
%(\cite{LeeS14,Madry16}) which is too slow for us, this still motivates our algorithm as follows. 
To really make an approximate flow method work in our dynamic setting, we need much more detailed knowledge of the pair of flow and cut we compute in addition to that they are approximately tight. Thus, we adapt the local flow method {\em Unit-Flow} by Henzinger, Rao, and Wang \cite{HenzingerRW17}, which is based on the Push-Relabel framework \cite{GT88}.

%\noindent
%\fbox{
%\parbox{0.97\columnwidth}{
\paragraph{Why push-relabel?}  
The reason behind this choice is that the flow and cut computed by {\em Unit-Flow} have certain nice invariants that make reusing flow routing across rounds very simple in terms of both operation and analysis. We need amortization in our analysis; the worst-case cost of one round maybe very high but they are fast amortized over all rounds. The potential function based analysis of {\em Unit-Flow} is very natural to adapt for the amortized analysis. 

It is conceivable that one can adapt blocking-flow based methods such as \cite{OrecchiaZ14} to be dynamic, but the two-level structure of Dinic's algorithm (i.e. multiple rounds of blocking flow computations) makes it more difficult to carry out the adaptation and running time analysis. In particular, blocking-flow based algorithms do not have the flexibility as push-relabel to extend naturally to amortized running time analysis.
 
Note that although we start with an undirected graph, if we want to reuse flow routing across rounds, it is natural to work with the residual network which is directed. Thus, it is not clear how to adapt approximate max flow algorithms designed for undirected graphs\footnote{whereas push-relabel based methods naturally work with directed graphs.} (\cite{KelnerLOS14,Sherman13}) to our dynamic setting.
%}}

%We show that reusing the routing efficiently is possible by using the {\em Unit-Flow} algorithm by Henzinger, Rao, and Wang \cite{HenzingerRW17}, which is a specific push-relabel based flow algorithm.
%%from \cite{HenzingerRW17}, and show how to reuse flow routing across rounds so the total running time is $O(\frac{|E(A,V-A)|\log m}{\phi^2})$ as stated in \ref{thm:trimming}. We defer this discussion to \ref{section:trimming}.
%It is conceivable that a variant of blocking-flow based algorithm by Orecchia and Zhu \cite{OrecchiaZ14} can be adapted as well.
%%Both algorithms are efficient by the same reason; they explicitly bound labels of nodes by a small number.
%We choose to use a push-relabel based algorithm because it is much simpler to analyze its running time. We need amortization in our analysis; the worst-case cost of one round maybe very high but they are fast amortized over all rounds. The potential function from push-relabel based algorithms is very flexible for the amortized analysis. Blocking-flow algorithms do not have such flexibility.
%%\footnote{Another high level explanation why push-relabel and blocking-flow algorithms with explicit bounds on the labels (e.g \cite{HenzingerRW17,OrecchiaZ14}) are efficient in our setting is from the observation that they indeed give $(1+\epsilon)$-approx max-flow on directed graphs with running time $\widetilde{O}(m/\epsilon)$ when each node $u$ is a sink of capacity $\deg(u)$.}
\subsubsection{Dynamic Adaptation of {\em Unit-Flow}}
\label{sec:dynamic unitflow}
Before we describe our dynamic adaptation of $\uf$, we need to define pre-flow.
%We give a quick overview of the {\em Unit-Flow} algorithm, which is a particular adaptation of the generic Push-Relabel algorithm. 
Given a flow problem $\Pi = (\Delta,T,c)$, a {\em pre-flow} $f$ is a feasible routing for $\Pi$ except the condition $\forall v: f(v) \leq T(v)$. We say that there is an \emph{excess} mass at $v$ is $f(v) > T(v)$.
We always consider flow problems where $\forall v:T(v)=\deg(v)$ and $\forall e: c(e)=2/\phi$ so we leave them implicit. 

Next, we just state some properties of $\uf$ that we need in this section (see \Cref{app:unit-flow} for the proof).
Given a graph $G=(V,E)$, a parameter $h$, and a source function $\Delta$ as inputs,
in the execution of {\em Unit-Flow}, it maintains a pre-flow $f$ and labels on nodes $l:V \rightarrow \{0,\dots,h\}$ with the following invariant:
\begin{enumerate}
	\item If $l(u) > l(v)+1$ where $(u,v)$ is an edge, then $(u,v)$ is saturated in the direction from $u$ to $v$, i.e. $f(u,v) = 2/\phi$.
	\item If $l(u) \ge 1$, then $u$'s sink is saturated, i.e. 
	$f(u) = \Delta(u) +  \sum_v f(v,u) \ge \deg(u)$.   
\end{enumerate}
We say that a tuple $(\Delta,f,l)$ is a {\em $G$-valid} state if it satisfies the invariant above. This notion captures what a intermediate state in the execution of {\em Unit-Flow} must obey (in addition to the conditions of a pre-flow). Furthermore, we say a tuple $(\Delta,f,l)$ is a {\em $G$-valid solution} if it satisfies the additional invariant
\begin{enumerate}
\setcounter{enumi}{2}
	\item If $l(u) < h$, then there is no excess mass at $u$, i.e. 
	$f(u) = \Delta(u) +  \sum_v f(v,u) \le \deg(u)$.   
\end{enumerate}
A $G$-valid solution is what {\em Unit-Flow} at termination must obey. Note by the last two invariants, we know $f(u)=\deg(u)$ if $1\le l(u)<h$ in a $G$-valid solution.

Normally, when we initialize {\em Unit-Flow}, the pre-flow $f$ is set as empty, i.e. $f(u,v)=0, \forall u,v$, and as well as the labels $l(u) =0, \forall u$. 
But {\em Unit-Flow} can in fact ``warm start'' on $G$ given any $G$-valid $(\Delta,f,l)$ since it is a valid intermediate state of the algorithm, and return a $G$-valid solution. We can use the pre-flow and node labels to find a cut. The lemma below makes this precise. As this can be derived directly from the analysis in \cite{HenzingerRW17}, we give the proof for completeness in \Cref{app:unit-flow}.
Let $|\Delta(\cdot)| = \sum_v \Delta(v)$ denote the total amount of initial mass. 
\begin{lem}\label{lem:key}
	Given a graph $G = (V,E)$ where $m = |E|$, a parameter $h$ and a $G$-valid state $(\Delta,f,l)$ where $|\Delta(\cdot)| \le 3m/2$,
%	\footnote{For efficiency, we assume that a list of nodes $u$ where $\ex_f(u)>0$ is also given as input. This can be easily handled because {\em Unit-Flow} already maintains this list.},
	{\em Unit-Flow} outputs a $G$-valid solution $(\Delta,f',l')$ 
%	where $\ex_{f'}(u)=0$ if $l'(u)<h$ 
	and a set $S \subseteq V$ where we must have one of the following two cases:
	\begin{enumerate}
		\item $S = \emptyset$ and $f'$ is a feasible flow for $\Delta$, or
		\item $S \neq \emptyset$ is a level cut, i.e., $S = \{u \in V \mid l'(u)\ge k\}$ for some $k\ge1$. Moreover, the number of edges $(u,v)$ such that $u\in S, v \in V\setminus S$ and $-2/\phi \leq f'(u,v) < 2/\phi$ is at most $\frac{5\vol(S)\ln 2m}{h}$. Note these edges include all the edges across the cut except those sending exactly $2/\phi$ units of mass {\em from $S$ into $V\setminus S$}.
	\end{enumerate}
%The algorithm takes $O(|\ex_f(\cdot)|h+\vol(S))$ time\footnote{For efficiency, we assume that a list of nodes $u$ where $\ex_f(u)>0$ is also given as input. This can be easily handled because {\em Unit-Flow} already maintains this list.}.
\end{lem}

\paragraph{The dynamic version.} 
%Recall \Cref{alg:trimming} for our trimming step. 
%When we apply \ref{lem:key} to the flow problem in any round of the trimming step, if we get case $(1)$, the trimming step can terminate since a feasible flow certifies the graph is a $\phi/6$ expander by \Cref{prop:certify}. If we get case $(2)$, the bound on the number of unsaturated edges across the cut $(S,V\setminus S)$ limit the amount of computation in the next round when we remove $S$, since for the new sources corresponding to the saturated cut edges, and $f'$ already routes $2/\phi$ units of mass from most of them to the part of the graph not removed.
%Given \Cref{lem:key}, we can describe the final trimming algorithm for \Cref{thm:trimming}.
Given $(G,A,\phi)$ as inputs from \Cref{thm:trimming}, we implement \Cref{alg:trimming} by using \Cref{lem:key} for finding a flow (and a cut) in each round. Between each round, we will dynamically update the underlying graph and the input parameters to \Cref{lem:key}.
It only remains to describe the parameters we feed to \Cref{lem:key} in each round. Instead of using the variable $A'$ as in \Cref{alg:trimming}, we use $A_1,A_2,\dots$ where $A_t$ denotes $A'$ in round $t$, and use $R$ to denote $\overline{A}$.

The parameter $h$ is fixed to be $h=\frac{40\ln 2m}{\phi}$ for every round. Initially, we set $A_{1}=A$. Set $f_{1}(u,v)=0$ and $l_{1}(u)=0$ for all $u,v\in A$. Set each edge from $E(A,R)$ as a source of $2/\phi$ units. Formally,
for each $u\in A$, $\Delta_{1}(u)=2/\phi\times|\{e\in E(A,R)\mid u\in e\}|$. 
Recall that each node $u$ is a sink of capacity $\deg(u)$, and each edge has capacity $2/\phi$. 
At round $t$, {\em Unit-Flow} is given as inputs $G\{A_t\}$ and $(\Delta_t,f_t,l_t)$ which is a $G\{A_t\}$-valid state, and then outputs a $G\{A_t\}$-valid solution $(\Delta_t,f'_t,l'_t)$ and a set $S_t$. These outputs from round $t$ is used to defined the inputs for the next round as follows.
\noindent
\fbox{
	\parbox{0.97\columnwidth}{
		{\bf Operations between round $t$ and $t+1$}
		\begin{enumerate}
			\item Let $(\Delta{}_{t},f'_{t},l'_{t})$ be the $G\{A_{t}\}$-valid solution
			and $S_{t}$ be the level cut outputted by \emph{Unit-Flow} at round
			$t$.
			\item Set $A{}_{t+1}\gets A_{t}-S_{t}$. 
			\item Let $f{}_{t+1}$ and $l{}_{t+1}$ be obtained from $f'_{t}$ and $l'_{t}$
			by restricting their domain from $A_{t}\times A_{t}$ to $A_{t+1}\times A_{t+1}$
			and from $A_{t}$ to $A_{t+1}$, respectively.
			\item \label{step: add mass}For each edge $e=(v,u)\in E(S_{t},A_{t+1})$
			where $u\in A_{t+1}$, set $e$ to be a source of $2/\phi$ units.
			More formally, for each $u\in A_{t+1}$, 
			\[
			\Delta_{t+1}(u)\gets\Delta_{t}(u)+\frac{2}{\phi}\times|\{e\in E(S_{t},A_{t+1})\mid u\in e\}|.
			\]
			
			\item Use $(\Delta{}_{t+1},f{}_{t+1},l{}_{t+1})$ as input for \emph{Unit-Flow}
			at round $t+1$.
		\end{enumerate}
}}

To digest the above process, consider any edge $e=(u,v)$ across the cut with $u\in S_t,v\in A_{t+1}$. In the flow problem of round $t+1$ in the trimming step, we need to route $2/\phi$ initial mass from $e$ to sinks in $A_{t+1}$. What we are doing between the rounds is indeed reusing the flow routing $f'_t$. Note the contribution of cutting $e$ to the change of mass on $v$\footnote{i.e. $e$'s contribution to $f_{t+1}(v)-f'_t(v)$} is exactly $2/\phi-f'_t(u,v)$. In particular, if $f'_t$ already routes $2/\phi$ mass from $u$ to $v$, we simply reuse the routing as if these are the $2/\phi$ units of initial mass out of $e$. If $0\leq f'_t(u,v) < 2/\phi$, then $f'_t$ only tells us how to route $f'_t(u,v)$ of the $2/\phi$ units of initial mass mass out $e$ into $A_{t+1}$, so we need to add the remaining mass on $v$ and let {\em Unit-Flow} further route these mass starting from the state $f_{t+1}$. When $f'_t(u,v)<0$, the pre-flow $f'_t$ actually routes mass in the wrong direction, that is, some of the initial mass from sources remaining in $A_{t+1}$ routes mass into $S_t$. However, in the flow problem of round $t+1$, we already removed $S_t$ from our graph, the routing of these mass given by $f'_t$ is no longer valid. Thus, we have to truncate the old routing of these mass at the cut edge $e$, and let {\em Unit-Flow} reroute these mass from $e$ back into $A_{t+1}$ (instead of into $S_t$ as in $f'_t$). Thus, the amount of mass we add to $v$ is more than $2/\phi$ to reflect the mass we need to reroute in addition to the $2/\phi$ initial mass we put on $e$ as a new source. 

The crucial fact from the above discussion is that the mass at a node $v$ only increases from the removal of a cut edge $(u,v)$ such that $f'_t(u,v) < 2/\phi$, and the increment of mass at $v$ is at most $4/\phi$ from one cut edge since $f'_t(u,v)\geq -2/\phi$. Note as we remove $S_t$ between the rounds, the mass on nodes in $S_t$ are removed along with $S_t$ forever, so we also ``destroy'' existing mass along the way. 

\paragraph{Analysis.} 
First, we show that the given parameters are applicable for the algorithm from \ref{lem:key}.
\begin{lem}
	For each round $t$, $(\Delta_{t},f_{t},l_{t})$ is a $G\{A_{t}\}$-valid state.\end{lem}
\begin{proof}
	We prove by induction. For round $1$, the tuple $(\Delta{}_{1},f{}_{1},l{}_{1})$
	is trivially $G\{A_{1}\}$-valid as $A_{1}=A$. For the inductive
	step, by \ref{lem:key}, we have that the output $(\Delta_{t},f'_{t},l'_{t})$
	is $G\{A_{t}\}$-valid. As $f_{t+1}$ and $l_{t+1}$ are just a restriction
	of $f_{t}$ and $l_{t}$, the first condition to being $G\{A_{t+1}\}$-valid
	holds. Indeed, for an edge $(u,v)$ where $u,v\in A_{t+1}$, if $l_{t+1}(u)>l_{t+1}(v)+1$,
	then $l'_{t}(u)>l'_{t}(v)+1$ and so $f_{t+1}(u,v)=f'_{t}(u,v)=2/\phi$.
	
	The second condition is a little tricky. In step $3$ when we restrict $f'_t$ to the domain $A_{t+1}\times A_{t+1}$, we actually also change the amount of mass on nodes. Since if a node $u\in A_{t+1}$ is adjacent to an edge $e\in(S,A_{t+1})$, when we go from $f'_t$ to $f_{t+1}$, the mass routed along $e$ into (or out of) $u$ is no longer counted when we compute mass on $u$. However, this reduces the mass at $u$ by at most $2/\phi$, and we add $2/\phi$ of mass as initial mass to $\Delta_{t+1}(u)$ for each cut edge, thus the amount of mass is non-decreasing for the nodes remaining in $A_{t+1}$.
	More formally, we have
	\begin{align*}
	f_{t+1}(u)= & \Delta_{t+1}(u)+\sum_{(v,u)\in E(A_{t+1},A_{t+1})}f_{t+1}(v,u)\\
	\ge & \Delta_{t}(u)+\sum_{(v,u)\in E(S_{t},A_{t+1})}f'_{t}(v,u)+\sum_{(v,u)\in E(A_{t+1},A_{t+1})}f'_{t}(v,u)\\
	= & \Delta_{t}(u)+\sum_{(v,u)\in E(A_{t},A_{t})}f'_{t}(v,u)
	= f'_{t}(u)\ge\deg(u).
	\end{align*}
	The inequality is due to the difference between $\Delta_{t+1}(u)$ and $\Delta_t(u)$ and $f'_t(v,u)\le 2/\phi$.	The second to last equality is because nodes incident to $u$ in $A_t$ must be either in $S_t$ or $A_{t+1}$. 
\end{proof}

Let $|\Delta_{total}(\cdot)|$ denote the total amount of mass we create over all rounds. 
More formally, for each round $t>1$, let us denote the {\em total new amount of mass at the beginning of round $t$} as 
$|\Delta_{t}^{new}(\cdot)|$.
The total amount of mass over all rounds includes the {\em initial} mass at the beginning of round $1$ and the {\em new} mass over all rounds $t>1$. That is, 
$$|\Delta_{total}(\cdot)| = |\Delta_{1}(\cdot)| + \sum_{t>1} |\Delta_{t}^{new}(\cdot)|.$$
Observe from the operations between round $t-1$ and $t$ that 
$$ |\Delta_{t}^{new}(\cdot)| = \sum_{(u,v) \in E(S_{t-1},A_t)} 2/\phi - f'_{t-1}(u,v)$$ where $u \in S_{t-1}$.

The next lemma is the key lemma for proving both correctness and running time
of the algorithm.
\begin{lem}
	$\vol(\cup_{t\ge1}S_{t}) \le |\Delta_{total}(\cdot)| \le  4|E(A,R)|/\phi$. \label{lemma:totalmass}
\end{lem}
\begin{proof}
	First, $|\Delta_{1}(\cdot)|=2|E(A,R)|/\phi$ is the amount
	of initial mass in the first round, as each cut edge in $E(A,R)$ is a source of $2/\phi$ units of mass. Next, for each round $t\ge1$, let $S_{t}$ be the
	level cut outputted from \ref{lem:key}. We will prove that before
	proceeding to round $t+1$: $(1)$ the amount of mass that we added is
	at most $\vol(S_{t})/2$, but $(2)$ the amount mass that we destroy is
	at least $\vol(S_{t})$. 
	This implies that $\vol(\cup_{t\ge1}S_{t})$ is upper-bounded by the total amount of destroyed mass, which is clearly at most the total amount of mass $|\Delta_{total}(\cdot)|$ created over all rounds. Moreover, for every unit of mass that we add, we destroy
	at least two units of mass. This allows us to charge every two units of mass we create to a distinct unit of existing mass, so $|\Delta_{total}(\cdot)| \le 2|\Delta_{1}(\cdot)| = 4|E(A,R)|/\phi$.
	
	Now, we prove $(1)$. 	By \ref{lem:key}, the number of edges where $v\in S_{t}$, $u\in A_{t+1}$,
	and $f'_{t}(v,u)<2/\phi$ is at most $\frac{5\vol(S)\ln 2m}{h}=\phi\vol(S_{t})/8$.
	For each such edge, the amount of new mass added is $2/\phi-f'_{t}(v,u)\le4/\phi$.
	So, we have the total amount of new mass at the beginning of round $t+1$ is 
	$|\Delta_{t+1}^{new}(\cdot)| \le  \frac{\phi\vol(S_{t})}{8}\times\frac{4}{\phi}=\vol(S_{t})/2$.
	For $(2)$, note $S_t$ is a level cut, so any $v\in S_t$ has $l'_t(v)>1$. As $(\Delta_{t},f'_{t},l'_{t})$ is valid, $f'_{t}(v)\ge\deg(v)$ by the second invariant of valid solutions
	for each $v\in S_{t}$. Once we remove $S_{t}$, all the mass on any $v\in S_t$ (either in the sinks or as excess mass)
	is gone forever, so we destroy at least $\vol(S_{t})$ units of mass.
\end{proof}

The running time analysis from \cite{HenzingerRW17} shows that {\em Unit-Flow} takes $O(|\Delta_{1}(\cdot)|h)$ in the first round. This analysis extends seamlessly to the dynamic version.
\begin{lem}\label{lemma:totaltime}
	The total running time of {\em Unit-Flow} in our trimming step over all rounds is $O(|\Delta_{total}(\cdot)|h) = O(|E(A,R)|\log m/\phi^2)$.
\end{lem}
%On a very high level, we can charged all operations of {\em Unit-Flow} to the mass we create over all rounds, and each unit of mass will be charged at most $O(h)$ in total. This is a direct consequence since {\em Unit-Flow} relabels any node at most $h$ times. Once a node reaches level $h$ in some round $t$, either we find a feasible flow in that round so the trimming step terminates, or the node is removed with the level cut $S_t$ since $S_t$ must contain all nodes in the highest level $h$. For readers familiar with the analysis of Push-Relabel, it comes as no surprise that our running time is proportional to $h$ as the analysis of Push-Relabel is centered around how many times a node can be relabeled. 
The formal proof requires understanding the implementation of {\em Unit-Flow} and is deferred to \Cref{app:unit-flow}.

\paragraph{Proof of \ref{thm:trimming}.} \Cref{lemma:totaltime} bounds the time. Let $A'=A-\cup_{t\ge1}S_{t}$. By \Cref{lemma:totalmass}, $\vol(A') \ge \vol(A) - 4|E(A,R)|/\phi$.
Suppose the algorithm finishes at round $T$, i.e. $A'=A_{T}$.
We get a feasible flow on $G\{A'\}$  where
every edge $E(A',V'-A)$ can send mass of $2/\phi$ units to
sinks in $G\{A'\}$.  But the total amount of mass is at most $4|E(A,R)|/\phi$.
So $E(A',V'-A)\le(\phi/2)\cdot4|E(A,R)|/\phi=2|E(A,R)|$. By \ref{prop:certify}, we also have $\Phi_{G\{A'\}}\ge\phi/6$.
\section{Discussion}
\label{section:discuss}
\subsection{Weighted Graphs}
\label{section:weighted}
Our results extend to weighted graphs in a straightforward manner. The cut-matching framework can be adapted to weighted graphs in a standard way. As to our dynamic flow subroutine, we need to use link-cut tree (i.e. dynamic tree) data structure for weighted graphs to make the running time proportional to the total number of edges rather than the total weight. Our dynamic flow subroutine in the trimming step works well with link-cut tree. The reason is that between rounds we remove level cuts, that is, we find some level $k$, and remove all nodes above the level $k$. For readers familiar with the application of dynamic tree in Push-Relabel algorithm, it is easy to see that to maintain the cut-link tree data structure across rounds, the removal of nodes in level cuts correspond to cut sub-trees in the data structure, which is easy to carry out. We state the results for weighted graphs, and refrain from giving a more tedious analysis.
\begin{thm}
	[Weighted Expander Decomposition]\label{thm:main-weighted}
	There is a randomized algorithm that with high probability given a graph $G=(V,E,w)$ with $m$ edges and total weight $W=\sum_e w_e= m^{O(1)}$ and a parameter
	$\phi$, partitions $V$ into $V_{1},\dots,V_{k}$ in
	time $O(m\log^{5}m/\phi)$ such that $\min_{i}\Phi_{G\{V_{i}\}}\ge\phi$
	and $\sum_{i}\delta(V_{i})=O(\phi W\log^{3}m)$\footnote{Here, volume and cut-size are defined as the sum of edges' weights rather than number of edges.}. 
	%Given weighted graph $G=(V,E,w)$, with $|E|=m$ and total weights $W=\sum_e w_e$ polynomially bounded by $m$. In $O(m\log^5 m/\phi)$ running time, we can find a partition of $V=V_1\uplus\ldots\uplus V_k$ for some $k$ such that each $G\{V_i\}$ is a $\Omega(\phi)$ expander (here volume is defined as the sum of edges' weights rather than number of edges):
	%\[
	%\forall i,\forall S\subset V_i, \vol_G(S)\leq \vol_G(V_i)/2: \sum_{e\in E(S,V_i-S)}w_e \geq \phi \vol_G(S)
	%\]
	%The total weight of edges between $V_i,V_j$ for $i\neq j$ is at most $O(W\phi\log^3m)$.
\end{thm}
%We proceed to discuss some of the main ideas of our algorithm, and we keep the discussion at high-level.
%Our result extends to weighted graphs in a fairly straightforward
%way by using the dynamic tree data structure. See {*}{*}{*} for more
%details.

\begin{thm}[Weighted Expander Pruning]
	\label{thm:pruning-weighted}Let $G=(V,E,w)$ be a $\phi$ expander with $m$ edges and total weight $W = \sum_e w_e$. There is a deterministic
	algorithm with access to adjacency lists of $G$ such that, given
	an online sequence of edge deletions in $G$ whose total weight of such edges is at most $\phi W/10$, can
	maintain a \emph{pruned set} $P\subseteq V$ such that the following
	property holds. Let $G_i$ and $P_{i}$ be the graph $G$ and the set $P$ after the $i$-th deletion. Let $W_i$ be the total weight of deleted edges up to time $i$.
	We have, for all $i$, 
	\begin{enumerate}
		\item $P_{0}=\emptyset$ and $P_{i}\subseteq P_{i+1}$,
		\item $vol(P_{i})\le8W_i/\phi$ and $|E(P_{i},V-P_{i})|\le4W_i$, and
		\item $G_i\{V-P_{i}\}$ is a $\phi/6$ expander. 
	\end{enumerate}
	The total time for updating $P_{0},\dots,P_{k}$ is $O(m\log m/\phi)$. 
\end{thm} 

It is impossible to make the amortized update time of \Cref{thm:pruning-weighted} hold for short update sequences. 
This is because the conductance can change drastically just by deleting an edge with very high weight.

\subsection{Applications}

\label{sub:app}

In this section, we list some applications which are implied by our
new expander decomposition and expander pruning in a fairly black-box manner.

%\paragraph{Spectral sparsifers. }
%
%As shown in \cite{SpielmanT04}, spectral sparsification in expanders
%can be obtained by simply random sampling each edge independently.
%\ref{thm:expander decomp} implies a very simple spectral sparsification
%algorithm for arbitrary graph; given any graph $G_{0}$ with $n$
%nodes and $m$ edge, we decompose $G$ into $\phi$ expanders where
%$\phi=O(1/\log^{3}n)$ and the remaining $m/2$ inter-cluster edges.
%Then, recurse on the graph induced by the inter-cluster edges. There
%are only $O(\log n)$ iterations, each of which gives a collection
%of expanders, each of which is trivial to sparsify. So this gives
%an $\widetilde{O}(m)$ time algorithm for obtaining $\widetilde{O}(n/\epsilon^{2})$-size
%$(1+\epsilon)$-spectral sparsification any graph. Note that spectral
%sparsification is widely used in undirected/directed Laplacian solvers
%\cite{SpielmanT04,CohenKPPRSV17}, max flow algorithms \cite{KelnerLOS14}.

\paragraph{User-friendly expander decomposition. }

The first construction of spectral sparsification of graphs by Spielman
and Teng \cite{SpielmanT04} is based on their expander decomposition.
As the clusters from their decomposition are only guaranteed to be
contained inside \emph{some} expander, this complicates the analysis of their sparsification algorithm.
\ref{thm:expander decomp} gives a more ``user-friendly'' decomposition
as each of our clusters induces an expander. By using \ref{thm:expander decomp},
both the algorithm for spectral sparsification and its analysis are
very simple\footnote{As shown in \cite{SpielmanT04}, spectral sparsification on expanders can be done by sampling each edge independently.}. 
%Given any graph
	%$G_{0}$ with $n$ nodes and $m$ edge, we decompose $G_{0}$ using
	%\ref{thm:expander decomp} into $\phi$ expanders where $\phi=O(1/\log^{3}n)$
	%and the remaining $m/2$ inter-cluster edges. Then, recurse on the
	%graph induced by the inter-cluster edges. There are only $O(\log n)$
	%iterations, each of which gives a collection of expanders, and are easy to sparsify. As the union of spectral sparsifiers
	%of disjoint subgraphs of $G$ is a sparsifier of $G$, we obtain an
	%$\widetilde{O}(m)$ time algorithm for obtaining $\widetilde{O}(n/\epsilon^{2})$-size
	%$(1+\epsilon)$-spectral sparsification any graph.}. 

The same complication arises in the analysis of several algorithms
which use Spielman and Teng's decomposition (e.g. spectral sketches
\cite{JambulapatiS18}, undirected/directed Laplacian solvers \cite{SpielmanT04,CohenKPPRSV17},
and approximate max flow algorithms \cite{KelnerLOS14}). By using
\ref{thm:expander decomp}, one can simplify the analysis of all these
algorithms.

\paragraph{Balanced low-conductance cut. }

If a graph $G=(V,E)$ has $\Phi_{G}<\phi$, then we can find in
$\widetilde{O}(m/\phi)$ a cut $(S,\overline{S})$ with conductance $O(\phi\log^{3}m)$
such that the volume of the smaller side $S$ is as large as any cut
$(T,\overline{T})$ with conductance at most $\phi$ up to a constant
factor. This is immediate from a $2\phi$ expander decomposition of
$G$; If all clusters have volume at most $9/10$
of the total volume, then we can easily group clusters to get a
cut $(S,\overline{S})$ where $\Phi(S)=\widetilde{O}(\phi)$ and $\min\{\vol(S),\vol(\overline{S})\}=\Omega(|E|)$.
Otherwise, there is one giant cluster $A$. Let $R=V-A$, and it is easy to see from the proof of \ref{thm:expander decomp} that
$\Phi(R)=\widetilde{O}(\phi)$. As $G\{A\}$ is a $2\phi$ expander, $R$
overlaps with half of the volume of \emph{every }low conductance cut,
i.e. for any cut $(S,\overline{S})$ where $\Phi(S)\le\phi$ and $\vol(S)\le\vol(\overline{S})$,
we have $\vol(S\cap R)\ge\vol(S)/2$. 

The previous best spectral-based algorithm by Orecchia, Sachdeva,
and Vishnoi \cite{OrecchiaSV12} has running time $\widetilde{O}(m)$
but the conductance of the output cut can be $\Omega(\sqrt{\phi})$.
The previous best flow-based algorithms have running time $\widetilde{O}(m)$, and combine the approximate max flow algorithm \cite{Peng14} with the framework
by Orecchia et al. \cite{OrecchiaSVV08} or Sherman \cite{Sherman09}.
These algorithms return a cut with conductance $O(\phi\log n)$ which
is better than ours. But the running time is $\Omega(m\log^{41}m)$ as they need approximate max flow\footnote{The exponent of $41$ may have been improved to some smaller (yet still much larger than $5$) number due to recent progress in approximate max flow.} which is much more complicated than ours.

%\paragraph{$(\phi,\epsilon)$-decomposition.}
%
%???\thatchaphol{may exclude for this version.}

\paragraph{Short cycle decomposition.}

An \emph{$(\hat{m},L)$-short cycle decomposition} of an undirected
graph $G$, decomposes $G$ into edge-disjoint cycles, each
of length at most $L$, with at most $\hat{m}$ edges not in these
cycles. Chu et al. \cite{ChuGPSSW18} give a short-cycle decomposition
construction using the expander decomposition of Nanongkai and Saranurak
\cite{NanongkaiSW17}. Given an $n$-node $m$-edge graph $G$, they
return $O(n^{1+O(1/\log^{1/4}n)},n^{O((\log\log n)^{3/4}/\log^{1/4}n)})$-short
cycle decomposition in time $m^{1+O(1/\log^{1/4}n)}$. Plugging
in our new expander decomposition immediately improves the result
to:
\begin{cor}
\label{cor:short cycle}There is an algorithm that, given a graph
$G$ with with $n$ vertices and $m$ edges, returns a $O(n^{1+O(\log\log n/\sqrt{\log n})},n^{O(\log\log n/\sqrt{\log n})})$-short
cycle decomposition of $G$ in time $m^{1+O(1/\sqrt{\log n})}$.
\end{cor}
Chu et al. \cite{ChuGPSSW18} show several applications of short cycle
decomposition including degree-preserving spectral sparsification,
sparsification of Eulerian directed graphs, and graphical spectral
sketches and resistance sparsifiers. \ref{cor:short cycle} implies
algorithms for constructing all these objects in the same running
time. 

\paragraph{Dynamic minimum spanning forests.}

\begin{comment}
Expander decomposition together with expander pruning are powerful
tools for dynamic graph algorithms where we want to maintain some
properties/structures of a graph undergoing updates. This is because,
in a very high level, 
\end{comment}

Expanders are well connected and are ``robust'' under edge deletions
(i.e. $k$ edge deletions can disconnect only $\widetilde{O}(k/\phi)$ volume
of a graph), and thus many dynamic graph problems become easier
once we assume that the underlying graph is always an expander (see
e.g. \cite{PatrascuT07}). A recent development on dynamic (minimum)
spanning trees \cite{NanongkaiS16,Wulff-Nilsen16a,NanongkaiSW17}
uses expander decomposition to decompose a graph into expanders and
expander pruning to maintain such an expander under edge deletions. With
this approach, they improved the long-standing $O(\sqrt{n})$ worst-case
update time on an $n$-node graph of \cite{Frederickson85,EppsteinGIN97}
to $n^{O(\log\log\log n/\log\log n)}=n^{o(1)}$. 

As we improve the running time of both expander decomposition and
pruning in \ref{thm:expander decomp} and \ref{thm:pruning} respectively,
we immediately improve the $o(1)$ in their running time as follows:
\begin{cor}
There is a Las Vegas randomized dynamic minimum spanning forest algorithm
on an $n$-node graph undergoing edge insertions and deletions with
$O(n^{O(\sqrt{\log\log n/\log n})})$ worst-case update time both
in expectation and with high probability.
\end{cor}

\paragraph{Dynamic low diameter decomposition.}

A $(\beta,D)$ low diameter decomposition of an $n$-node $m$-edge
graph $G=(V,E)$ is a partition $V_{1},\dots,V_{k}$ of $V$ into
clusters where, for each $i$, the induced subgraph $G[V_{i}]$ has
diameter at most $D$ and the number of inter-cluster edges is at
most $\beta m$. Goranci and Krinninger \cite{GoranciK18} show 
a dynamic algorithm for maintaining a low diameter decomposition which
implies dynamic low-stretch spanning trees as an application.

From our new expander decomposition and pruning algorithm, we can
maintain low diameter decomposition with almost the same guarantees
as in \cite{GoranciK18} (up to a small polylogarithmic factor). But
we do not need to assume that the adversary is \emph{oblivious }(i.e.
the adversary fixes the whole sequence of updates from the beginning)\footnote{This result, however, does not imply dynamic low-stretch spanning
tree algorithm without the oblivious adversary assumption, because
there are other steps which need this assumption.}. 
\begin{cor}
\label{cor:low stretch}Given any unweighted, undirected multigraph
undergoing edge insertions and deletions, there is a fully dynamic
algorithm for maintaining a $(\beta,\widetilde{O}(1/\beta))$ low diameter
decomposition that has amortized update time $\widetilde{O}(1/\beta^{2})$.
The amortized number of edges to become inter-cluster edges after
each update is $\widetilde{O}(1/\beta)$. These guarantees hold for a \emph{non-oblivious} adversary.
\end{cor}
While many randomized dynamic algorithms in the literature assume
an oblivious adversary (e.g. \cite{BaswanaGS11,BaswanaKS12,KapronKM13,HenzingerKNFOCS14,HenzingerKNstoc14,GoranciK18}),
there are very few techniques in dynamic graph algorithm for removing
the oblivious adversary assumption. This is important for some applications
and is also a stepping stone for derandomizing such algorithms. \ref{cor:low stretch}
can be seen as a progress in this research program.

%\subsection{Organization}
%....

\subsection{Open problems}
Derandomizing the expansion decomposition in \Cref{thm:expander decomp} is very interesting as it would break a long-standing $O(\sqrt{n})$ worst-case
update time of deterministic dynamic connectivity on $n$-node graphs
to $n^{o(1)}$ worst-case update time by using the framework of \cite{NanongkaiSW17}. An efficient parallel or distributed implementation of our expander decomposition algorithm will also be of interesting. Chang, Pettie, and Zhang \cite{ChangPZ18} recently give a distributed algorithm of a weaker variant of expander decomposition with some application.

For the trimming step, can we remove the dependency on $1/\phi$ in the running time? This will very likely require some dynamic version of the nearly linear time approximate max flow algorithms~\cite{KelnerLOS14,Sherman13}. To what extent can we reduce the number of inter cluster edges? Improving ours by a logarithmic factor to $O(\phi m \log^2 m)$ might be possible using a more complicated variant of the cut-matching game \cite{OrecchiaSVV08}. 
A challenging question about expander pruning algorithm in \Cref{thm:pruning} is whether the amortized update time of $O(\log n/\phi^{2})$ can be made worst-case\footnote{Currently, using \Cref{thm:pruning} and the technique from Section 5.2 of \cite{NanongkaiSW17}, we can obtain a worst-case algorithm with bad update time and guarantee. More precisely, when initially $\phi = \Omega(1/\text{polylog}(n))$, the update time is $2^{O(\sqrt{\log n})}$. More importantly, it only guarantees that the set $P$ contains some set $W$ where $G\{V-W\}$ is a $1/2^{O(\sqrt{\log n})}$ expander, instead of the complement $G\{V-P\}$ itself}.

\section*{Acknowledgement}	
	We thank Satish Rao, Richard Peng, Julia Chuzhoy and anonymous reviewers for helpful suggestion. 
	Wang is supported by NSF award CCF 1718533.
	This project has received funding from the European Research Council (ERC) under the
	European Union’s Horizon 2020 research and innovation programme under grant agreement
	No 715672. Saranurak was also partially supported by the Swedish Research
	Council (Reg. No. 2015-04659). 
	
\bibliographystyle{alpha}
\bibliography{references}
\begin{appendix}
	\section{Details of \emph{Unit-Flow}}

\label{app:unit-flow} 

In this section, we describe our adaptation of the $\uf$ algorithm from \cite{HenzingerRW17},
which in turn is based on the Push-Relabel framework by Goldberg and Tarjan \cite{GT88}.
The way we apply $\uf$ is different from the way used in \cite{HenzingerRW17} in two aspects.
First, in \cite{HenzingerRW17}, each node has initial mass
at most twice its degree. Here, we may each node $u$ might have $(2/\phi)\cdot\deg(u)$
units of initial mass. 
Second, and more importantly, we need to analyze the running time of $\uf$ over many calls where the underlying graph is dynamically updated.
However, the analysis extends seamlessly. We show the algorithm and analysis for completeness here.

\paragraph{Preflow.}

Recall the definitions about flow from \ref{sec:flow_prelim}. Given
a flow problem $\Pi=(\Delta,T,c)$, a \emph{pre-flow} $f$ is a feasible
routing for $\Pi$ except the condition $\forall v:f(v)\leq T(v)$.
As pre-flow may not obey sink capacity on nodes, we define the \emph{absorbed}
mass on a node $v$ as $\ab_{f}(v)=\min(f(v),T(v))$. We have $\ab_{f}(v)=T(v)$
iff $v$'s sink is saturated. The \emph{excess} on $v$ is $\ex_{f}(v)=f(v)-\ab_{f}(v)$.
From the definition, when there is no excess, $\forall v:\ex_{f}(v)=0$,
then $f$ is a feasible flow for $\Pi$. We omit the subscript whenever
it is clear. From now, we always consider flow problems where $\forall v:T(v)=\deg(v)$
and $\forall e:c(e)=2/\phi$ so we leave them implicit. 

\paragraph{Push-Relabel framework.}

In the generic Push-Relabel framework, each node $v$ has a non-negative
integer label $l(v)$, which starts at $0$ and only increases throughout
the algorithm, and we say a node $v$ is at level $i$ if $l(v)=i$.
The standard residual network is also maintained, where each undirected
edge $\{u,v\}$ corresponds to two directed arcs $(u,v),(v,u)$, and
the residual capacity of an arc $(u,v)$ is $r_{f}(u,v)=c(\{u,v\})-f(u,v)$.
The label is maintained such that if $r_{f}(u,v)>0$, then $l(u)\le l(v)+1$.
An arc $(u,v)$ is \emph{eligible} if $r_{f}(u,v)>0$ and $l(u)=l(v)+1$. 

Push-Relabel-based algorithm works as follows. Initially, $l(v)=0$
for all $v$ and $f(u,v)=0$ for all $u,v$. Then, the generic algorithm
picks any node $u$ with excess mass (i.e. $\ex(u)>0$), and tries
to push the excess mass away from $u$ along eligible arcs while respecting
the residual capacity. If there is no eligible arcs adjacent to $u$
to push the excess, the algorithm raises $l(u)$ by $1$. The algorithm
keeps doing this until there is no excess on any node. 

\paragraph{$\protect\uf$.}

Given a graph $G=(V,E)$, a parameter $h$, and a source function
$\Delta$ as inputs, the $\uf$ algorithm follows the Push-Relabel
framework with two adaptations. First, the algorithm stops trying
to push the excess mass from $u$ once the node $u$ is raised to level
$h$ where $h$ is a given parameter. Formally, a node $v$ is \emph{active}
if $\ex(v)>0$ and $l(v)<h$, and we maintain a queue $Q$ of active
nodes. The algorithm only pushes from active nodes. Second, the algorithm
enforces the excess on each node $u$ to be at most \emph{$\deg(u)$
	unless} that excess is initially placed at $u$. Formally, a unit
of excess on a node $u$ is called an\emph{ initial excess} if the
excess is initially placed on $u$ by $\Delta(\cdot)$ and is never
moved by our algorithm (it is helpful to think of mass as distinct
discrete tokens). We will enforce that if any excess on a node $u$
is not initial excess, then $ex(u)\leq\deg(u)$. See \ref{alg:unit-flow}
for the formal description of the algorithm.

\begin{algorithm}
	\caption{Unit Flow}
	\label{alg:unit-flow}
	\fbox{
		\parbox{0.97\columnwidth}{
			{\em Unit-Flow}($G$,$h$,$(\Delta,f,l)$) \\ 
			%\tab {\bf Initialization:}\\
			%\tab \tab $\forall \{v,u\}\in E$, $f(u,v)=f(v,u)=0$.\\
			%\tab \tab $Q=\{v|\Delta(v)>\deg(v)\}$, $\forall v$, $l(v)=0$.\\
			\tab {\bf Assertion}: $(\Delta,f,l)$ is $G$-valid.\\
			\tab {\bf While} $\exists v$ where $l(v)<h$ and $\ex(v)>0$\\
			\tab \tab Let $v$ be such vertex with smallest $l(v)$.\\
			\tab \tab {\em Push/Relabel}$(v)$. \\
			\tab {\bf End While}}}
	\fbox{
		\parbox{0.97\columnwidth}{
			{\em Push/Relabel}$(v)$ \\
			\tab {\bf If} $\exists$ arc $(v,u)$ such that $r_f(v,u)>0$, $l(v)=l(u)+1$\\
			\tab \tab {\em Push}$(v,u)$.\\
			\tab {\bf Else} {\em Relabel}$(v)$.\\
			\tab {\bf End If}
	}}
	\fbox{
		\parbox{0.97\columnwidth}{
			{\em Push}$(v,u)$\\
			\tab {\bf Assertion}: $\ex(u)=0$.\\
			\tab $\psi = \min\left(\ex(v),r_f(v,u),\deg(u)-\ex(u)\right)$\\
			\tab Send $\psi$ units of supply from $v$ to $u$:\\
			\phantom{\tab} $f(v,u)\leftarrow f(v,u)+\psi, f(u,v)\leftarrow f(u,v)-\psi$.
	}}
	\fbox{
		\parbox{0.97\columnwidth}{
			{\em Relabel}$(v)$\\
			\tab {\bf Assertion}: $v$ is active, and $\forall u\in V$\\
			\phantom{\tab {\bf Assertion}:} $r_f(v,u)>0\implies l(v)\leq l(u)$.\\
			\tab $l(v)\leftarrow l(v)+1$.
	}}
\end{algorithm}

To implement \ref{alg:unit-flow}, we need to specify how we maintain
the list $Q$ of active vertices. We keep all active vertices in non-decreasing
order with respect to their labels (breaking ties arbitrarily). Thus,
when we access the first vertex $v$ in $Q$, $v$ is the active node
with smallest label, and if Push$(v,u)$ is called, the assertion
$\ex(u)=0$ must be satisfied, as otherwise $u$ will have excess
on it (i.e. active) and has smaller label. It is easy to see that
adding a node to $Q$, removing a node from $Q$, or moving the position
of a node in $Q$ due to relabel can all be carries out in $O(1)$
time using one linked list for each label value. 

In \cite{HenzingerRW17}, we initialize $l(v)=0$ for all $v$ and
$f(u,v)=0$ for all $u,v$ as standard. The lemma below explicitly
states the properties of $f$ and $l$ maintained by $\uf$.
\begin{lemma}
	\label{lem:invariant unitflow}During the execution, $\uf$ maintains
	a pre-flow $f$ and labels on nodes $l:V\rightarrow\{0,\dots,h\}$
	so that $\ensuremath{(\Delta,f,l)}$ is \emph{$G$-valid}, given that
	$\ensuremath{(\Delta,f,l)}$ was $G$-valid at the beginning. That
	is, the following invariants are maintained: 
	\begin{enumerate}
		\item If $l(u)>l(v)+1$ where $(u,v)$ is an edge, then $(u,v)$ is saturated
		in the direction from $u$ to $v$, i.e. $f(u,v)=2/\phi$. 
		\item If $l(u)\ge1$, then $u$'s sink is saturated, i.e. $\ab_{f}(v)=\deg(u)$. 
	\end{enumerate}
	At termination of $\uf$, $\ensuremath{(\Delta,f,l)}$ a \emph{$G$-valid
		solution.} That is, if $l(u)<h$, then there is no excess mass at
	$u$, i.e. $f(u)=\Delta(u)+\sum_{v}f(v,u)\le\deg(u)$. 
\end{lemma}

\begin{proof}
	The first invariant follows from the standard proerty of a standard
	property of the push-relabel framework, where $r_{f}(u,v)>0$ implies
	$l(u)\leq l(v)+1$. For the second property, observe that the amount
	of absorbed mass at $v$, $\ab_{f}(v)$, is maintained by $\uf$ so
	that it is non-decreasing. Thus any time after the point that $v$
	first becomes active, the sink at $v$ is saturated. At termination,
	there is no excess at $u$ unless $l(u)=h$. Hence, the last statement
	follows. 
\end{proof}

Now, we are ready to prove several lemmas which are deferred from the main body of the paper.
\subsection{Proof of \ref{lem:key}}
Let $f'$ and $l'$ be the pre-flow and labels at termination of $\uf$.
We use the labels at the end of $\uf$ to divide the vertices into
groups 
\[
B_{i}=\{v|l'(v)=i\}
\]
If $B_{h}=\emptyset$, no vertex has positive excess, and we end up
with case $(1)$.

As $|\Delta(\cdot)|\leq3m/2$, we only have enough mass to saturate
a region of volume $3m/2$, thus $\vol(B_{0})>m/2$. We consider level
cuts as follows: Let $S_{i}=\cup_{j=i}^{h}B_{j}$ be the set of vertices
with labels at least $i$. For any $i$, an edge $\{v,u\}$ across
the cut $S_{i}$, with $u\in S_{i},v\in V\setminus S_{i}$, must be
one of the two types: 
\begin{enumerate}
	\item $l'(u)=i,l'(v)=i-1$. 
	\item $l'(u)-l'(v)>1$. Note that, by the first invariant of \ref{lem:invariant unitflow},
	all edges of this type must be saturated in the down-hill direction,
	i.e. $f'(u,v)=2/\phi$. 
\end{enumerate}
Suppose thesre are $z_{1}(i)$ edges of the first type, we sweep from
$h$ to $1$, and there must be some $i^{*}\in[1,h]$ such that 
\begin{equation}
z_{1}(i^{*})\leq\frac{5\vol(S_{i^{*}})\ln m}{h}.\label{eq:eligible}
\end{equation}
We let $S$ be $S_{i^{*}}$, and as $S$ is a level cut, $S$ clearly
satisfies all the properties in case $(2)$ of the lemma. We only
need to argue its existence. By the following region growing argument,
we can show there must exist such $i^{*}\geq1$. We start the region
growing argument from $i=h$ down to $1$. By contradiction, suppose
$z_{1}(i)\geq\frac{5\vol(S_{i})\ln m}{h}$ for all $i\in[1,h]$, which
implies $\vol(S_{i})\geq\vol(S_{i+1})(1+\frac{5\ln m}{h})$ for all
$hi\in[1,h)$. Since $\vol(S_{h})=\vol(B_{h})\geq1$, we will have
$\vol(S_{1})\geq(1+\frac{5\ln m}{h})^{h}\gg2m$, which contradicts
$\vol(S_{1})\leq3m/2$.

If we need to compute a level cut $S$, the sweep cut procedure takes
$O(\vol(S))$ since we stop as soon as we find a satisfiable $S$. 

\subsection{Warm-up proof of \ref{lemma:totaltime}}

Before proving \ref{lemma:totaltime} which is about the total running
time of $\uf$ over many calls to the algorithm, we first look at
the simplified case where we run $\uf$ only once with initial mass
given by $\Delta(\cdot)$ and label bound $h$. The proof is almost
identical to the one in \cite{HenzingerRW17} except that, here, the amount of initial mass on $v$ can exceed $2\deg(v)$, i.e.,
possibly $\Delta(v)>2\deg(v)$.
\begin{lemma}
	\label{lemma:unitflow-runtime} The running time of $\uf$ is $O(|\Delta(\cdot)|h)$. 
\end{lemma}

\begin{proof}
	We initialize the list of active vertices $Q$ in time linear in $|\Delta(\cdot)|$.
	For the work carried out by $\uf$, we will first charge the operations
	in each iteration of $\uf$ to either a push or a relabel. Then we
	will in turn charge the work of pushes and relabels to the mass we
	have, so that each unit of mass gets charged $O(h)$ work. This will
	prove the result, as there are $|\Delta(\cdot)|$ units of mass in
	total.
	
	In each iteration of {\em Unit-Flow}, we look at the first element
	$v$ of $Q$, suppose $l(v)=i$ at that point. If the call to {\em
		Push/Relabel}$(v)$ ends with a push of $\psi$ units of mass, the
	iteration takes $O(\psi)$ total work. Note we must have $\psi\geq1$
	by the assertion of Push$(v,u)$. We charge the work of the iteration
	to that push. Otherwise, the call to {\em Push/Relabel}$(v)$ doesn't
	push and, hence, relabels $v$. We observe that if we cannot push
	along an arc $(v,u)$ at some point, we won't be able to push along
	the arc until the label of $u$ raises. Thus, for any fixed label
	value of $v$, we read the arc $(v,u)$ without doing push only once.
	Therefore, we can change $O(\deg(v))$ to each relabel of $v$.
	
	So far we have charged all the work to pushes and relabels, such that
	a push of $\psi$ units of mass takes $O(\psi)$, and each relabel
	of $v$ takes $O(\deg(v))$. We now charge the work of pushes and
	relabels to the mass. %
	\begin{comment}
	We consider the absorbed mass at $v$ as the first up to $\deg(v)$
	units of mass started at or pushed into $v$, and these units never
	leave $v$. Recall that we call a unit of excess on $u$ initial excess
	if it started at $u$ as source mass and never moved.
	\end{comment}
	By \ref{lem:invariant unitflow}, each time we relabel $v$, we have
	$ex_{f}(v)>0$. So $\ab_{f}(v)=\deg(v)$ and each absorbed unit gets
	charged $O(1)$. A vertex $v$ is relabeled at most $h$ times, so
	each unit of mass, as absorbed mass, is charged with $O(h)$ in total
	by all the relabels.
	
	For the pushes, we consider the potential function 
	\[
	\Lambda=\sum_{v}\ex(v)l(v)
	\]
	Each push operation of $\psi$ units of mass decreases $\Lambda$
	by exactly $\psi$, since $\psi$ units of excess mass is pushed from
	a vertex with label $i$ to a vertex with label $i-1$. $\Lambda$
	is always non-negative, and it only increases when we relabel some
	vertex $v$ with $\ex(v)>0$. When we relabel $v$, $\Lambda$ is
	increased by $\ex(v)$. If $\ex(v)>\deg(v)$, all the excess on $v$
	must be initial excess of $v$. In this case, we charge $O(1)$ to
	each unit of excess on $v$, and in total each unit of mass, as initial
	excess, can be charged $O(h)$. If $\ex(v)\leq\deg(v)$, we can charge
	the increase of $\Lambda$ to the $\deg(v)$ absorbed mass at $v$,
	and each unit gets charged with $O(1)$. In total we can charge all
	pushes (via $\Lambda$) to the mass we have, and each unit of mass
	is charged with $O(h)$ in total as absorbed mass or initial excess. 
\end{proof}

\subsection{Proof of \ref{lemma:totaltime}}

Now we prove the running time when we run {\em Unit-Flow} over
rounds in the trimming step (\ref{lemma:totaltime}). %\totaltime*

\begin{proof}
	The proof is very similar to the proof of \ref{lemma:unitflow-runtime},
	and we only point out how to change the argument in proof of \ref{lemma:unitflow-runtime}
	to accommodate the mass added between rounds.
	
	We can still charge all operations to pushes and relabels, and then
	charge the pushes to the increment of the potential function 
	\[
	\Lambda=\sum_{v}\ex(v)l(v)
	\]
	The only part that differs from \ref{lemma:unitflow-runtime} is that
	apart from relabels, we may also increase $\Lambda$ when we add initial
	mass to nodes between rounds. In this case, $\Lambda$ can be increased
	by at most the amount of new mass added multiplied by the label of
	the node where we add the mass to, thus we can charge each unit of
	added mass at most $O(h)$ for the increment of $\Lambda$. Each unit
	of mass is only charged once this way when it is added.
	
	Same as in the proof of \ref{lemma:unitflow-runtime}, we can then
	charge all the work to mass, such that each unit of mass is charged
	with $O(h)$. Recall the definition of $|\Delta_{total}(\cdot)|$, the total amount of mass we create across all rounds, defined in \Cref{sec:dynamic unitflow}.
	The total running time is then $O(|\Delta_{total}(\cdot)|h)$. The
	time to compute the cut $S$ by sweep cut between rounds takes $O(\vol(S))$,
	and because the total volume of all the cuts we remove is bounded
	by $|\Delta_{total}(\cdot)|$ by \ref{lemma:totalmass}, so the total
	running time for computing such cuts is $O(|\Delta_{total}(\cdot)|)$.
\end{proof}

\subsection{Proof of \Cref{thm:pruning}: expander pruning}

Expander pruning is a simple extension of an algorithm for the trimming
step. So similarly, we first prove a proposition analogous to \ref{prop:certify}.
\begin{prop}
	\label{prop:certify pruning}Let $G=(V,E)$ be a graph where $\Phi_{G}\ge\phi$.
	Let $D\subseteq E$ be a subset of edges and $G'=G-D$. For any $A\subseteq V$,
	if $G'\{A\}$ is not a $\phi/6$ expander, then the following flow
	problem in $G\{A\}$ is not feasible: 
	\begin{itemize}
		\item each node $v$ is a sink of capacity equal to its degree $\deg(v)$, 
		\item edges all have capacity $2/\phi$,
		\item each edge in $E(A,V-A)$ is a source of $2/\phi$ units. For each
		end point $u$ of edges in $D$, add an initial mass of $4/\phi$
		units at $u$.
	\end{itemize}
\end{prop}

\begin{proof}
	If $G'\{A\}$ is not a $\phi/6$ expander, then there is an $S\subset A$
	where $\vol(S)\le\vol(A)/2$ such that $\frac{|E_{G'}(S,A-S)|}{\vol_{G'}(S)}\le\phi/6$.
	For any $u$, let $\deg_{D}(u)=|\{(u,v)\in D\}|$ and $\vol_{D}(S)=\sum_{u\in S}\deg_{D}(S)$.
	So we have 
	
	\[
	\frac{|E(S,A-S)|-\vol_{D}(S)}{\vol(S)}\le\phi/6.
	\]
	On the other hand, as $\Phi_{G}\ge\phi$, we have 
	\[
	\frac{|E(S,V-S)|}{\vol(S)}\ge\phi.
	\]
	So $6(|E(S,A-S)|-\vol_{D}(S))\le|E(S,V-S)|=|E(S,A-S)|+|E(S,V-A)|$.
	So 
	\[
	5|E(S,A-S)|\le|E(S,V-A)|+6\vol_{D}(S)
	\]
	and we have
	
	\begin{align*}
	\vol(S) & \le\frac{|E(S,V-S)|}{\phi}\\
	& =\frac{|E(S,A-S)|+|E(S,V-A)|}{\phi}\\
	& =\frac{\frac{6}{5}|E(S,V-A)|+\frac{6}{5}\vol_{D}(S)}{\phi}
	\end{align*}
	Now, observe that the total amount of source in $S$ in the flow problem
	is $\frac{2}{\phi}|E(S,V-A)|+\frac{4}{\phi}\vol_{D}(S)$. But the
	total capacity of flow that can absorb or goes out of $S$ is at most
	\begin{align*}
	\vol(S)+\frac{2}{\phi}E(S,A-S) & \le\frac{\frac{6}{5}|E(S,V-A)|+\frac{6}{5}\vol_{D}(S)}{\phi}+\frac{\frac{2}{5}|E(S,V-A)|+\frac{2\cdot6}{5}\vol_{D}(S)}{\phi}\\
	& =\frac{\frac{8}{5}|E(S,V-A)|+\frac{18}{5}\vol_{D}(S)}{\phi}\\
	& <\frac{2}{\phi}|E(S,V-A)|+\frac{4}{\phi}\vol_{D}(S).
	\end{align*}
	So we conclude that this flow problem is not feasible.
\end{proof}

\paragraph{The algorithm.}

This motivates the following algorithm. Given an initial graph $G$
where $\Phi_{G}\ge\phi$, we assign each node $v$ as a sink of capacity
equal to its degree $\deg(v)$ and edges all have capacity $2/\phi$.
We will maintain a set $A\subset V$ and a $G\{A\}$-valid $(\Delta,f,l)$
where $\Delta$ is a source function, $f$ is a preflow, and $l$
is a level function. Initially, $\Delta(u),f(u,v),l(u)=0$ for all
$u,v$. 

Given an edge $e=(u,v)$ to be deleted, set $\Delta(u)\gets\Delta(u)+4/\phi$
and $\Delta(v)\gets\Delta(v)+4/\phi$. Observe that $(\Delta,f,l)$
is still $G\{A\}$-valid given that it was before. Then, we run $\uf$
on $G\{A\}$ with the $G\{A\}$-valid tuple $(\Delta,f,l)$. As in
the trimming step, we will for many rounds execute $\uf$ and adjust
$(\Delta,f,l)$ until we obtain $A'\subseteq A$ and $(\Delta',f',l')$
which is $G\{A'\}$-valid and $f'$ is feasible flow for the source
function $\Delta'$. Note that the source function $\Delta'$ is as
described in \ref{prop:certify pruning} where $D$ is the set of
all deleted edge so far. Hence, we have that $G'\{A'\}$ a $\phi/6$
expander where $G'=G-D$.

The answer of the expander pruning algorithm is the set $P=V-A'$.
We update $A\gets A'$, $f\gets f'$, and $l\gets l'$. Then do the
same given the next update. Note that the algorithm works on the graph
$G\{A\}$ and not $G'\{A\}$ whose edges are deleted. We only update
the source function according to the deleted edges. 
This lemma below concludes the proof of \ref{thm:pruning}.
\begin{lem}
	We have the following:
	\begin{enumerate}
		\item At any time, $|\Delta_{total}(\cdot)|\le8|D|/\phi$. 
		\item $\vol(P)\le|\Delta_{total}(\cdot)|\le8|D|/\phi$.
		\item $E(P,V-P)\le\frac{\phi}{2}|\Delta_{total}(\cdot)|\le4|D|$.
		\item The total running time is at most $O(|\Delta_{total}(\cdot)|h)=O(|D|\log n/\phi^{2})$.
	\end{enumerate}
\end{lem}

\begin{proof}
	For (1), the total amount of mass added by the edge deletions is clearly
	$4|D|/\phi$. By the same analysis as in \Cref{lemma:totalmass}, the
	total amount of mass added between each round of the executing of
	$\uf$ is at most $4|D|/\phi$. For (2)-(4), these follow using the
	same analysis as in the proof of \Cref{lemma:totalmass}, \Cref{thm:trimming} and \Cref{lemma:totaltime}, respectively.
\end{proof}

	\section{The Cut-Matching Step}
\label{section:cut-matching}

The goal of this section is to prove \ref{thm:cut-matching}. We will work on the \emph{subdivision graph}\footnote{The reason we work on this graph is that, technically, the cut-matching
	framework is for certifying that a graph high \emph{expansion} or
	finding a cut with low expansion. Expansion is slightly different
	from conductance. An expansion of a cut $S$ in a graph $G=(V,E)$
	is $h(S)=\delta(S)/\min\{|S|,|V-S|\}$. An expansion of a graph $G$
	is $h_{G}=\min_{\emptyset\neq S\subset V}h(S)$.} of $G=(V,E)$ defined as follows:
\begin{definition}[Subdivision Graph]\label{def:subdivision}
	Given a graph $G=(V,E)$, its subdivision graph $G_E=(V',E')$ is the graph where we put a split node $x_e$ on each edge $e\in E$ (including the self-loops). Formally, $V'=V\cup X_E$ where $X_E=\{x_e|e\in E\}$, and $E'=\{\{u,x_e\},\{v,x_e\}|e=\{u,v\}\in E\}$. We will call nodes in $X_E$ the split nodes, and the other nodes in $V'$ the regular nodes.
\end{definition}
Let us define \emph{$X_{E}$-commodity flow} as a multi-commodity
flow on $G_{E}$ such that there are $|X_{E}|$ flow commodities
and every split note $x_{e}\in X_{E}$ is a source of quantity $1$
of its distinct flow commodity. Only for analysis, we consider an
$m\times m$ \emph{flow-matrix} $F\in\mathbb{R}_{\ge0}^{E\times E}$
which encodes information about an $X_{E}$-commodity flow. For
any two split nodes $x_{e}$ and $x_{h}$, $F(x_{e},x_{h})$ indicates
how much $x_{h}$ receives the flow commodity from $x_{e}$. 
We say that $F$ is \emph{routable with congestion $c$}, if
there exists a $X_{E}$-commodity flow $f$ is such that, simultaneously for
every $x_{e}$ and $x_{h}$, we have that $x_{e}$ can send $F(x_{e},x_{h}$) flow
commodity to $x_{h}$, and the amount of flow on each edge is at most
$c$.

\begin{algorithm}
	\caption{Cut Matching}
	\label{alg:cut-matching}
	\fbox{
		\parbox{0.97\columnwidth}{
			{\em Cut-Matching}($G$,$\phi$) \\ 
			\tab Construct subdivision graph $G_E$\\
			\tab $A=V\cup X_E,R=\emptyset,T=\Theta(\log^2 m)$,\\
			\phantom{\tab} $t=1,c=1/(\phi T)$.\\
			\tab Implicitly set $F$ as the identity matrix of size $m\times m$.\\
			\tab {\bf While} $\vol(R)\leq m/10T$ and $t\leq T$\\
			\tab \tab $t = t+1$\\
			\tab \tab Implicitly update $F$ so that $F$\\
			\phantom{\tab \tab} is still routable with congestion $ct$.\\
			\tab \tab This update operation might return $S\subset A$\\
			\phantom{\tab \tab} where $\Phi_{G_E\{A\}}(S) \le 1/c$.\\
			\tab \tab {\bf If} $S$ is returned\\
			\tab \tab \tab $R=R\cup S, A=A-S$
	}}
\end{algorithm}

Our algorithm is described in \ref{alg:cut-matching}. We are given $G=(V,E)$ with $m$ edges and $\phi$. We assume $\phi<1/(\log^2 m)$, otherwise the theorem is trivial. Initially,
we set $A=V\cup X_{E},R=\emptyset,T=\Theta(\log^{2}m)$ and $c=1/(\phi T)$. Only for analysis, we implicitly set $F$ as the identity matrix corresponding to the $X_{E}$-the commodity flow where each split node $x_{e}$ has 1 unit of its own commodity. Trivially, $F$ is routable with zero congestion.
Then, we proceed in for at most $T$ rounds and stop as soon as $\vol(R)>m/10T=\Omega(m/\log^{2}m)$.
For each round $t$, 
we will {\em implicitly} update $F$ such that it is routable with congestion at most $ct$.
The operation for updating $F$ will be described explicitly later in \Cref{section:cut-matching-proof}. 
If such operation returns a cut $S\subset A$ where $\Phi_{G_{E}\{A\}}(S)\le1/c$,
then we ``move'' $S$ from $A$ to $R$, i.e., $\ensuremath{R=R\cup S}$
and $A=A-S$. 
\begin{rem}
	\label{remark:subdivision}
	We make several remarks for our following discussion.
	\begin{enumerate}
		\item We only consider subsets of vertices $S$ in the subdivision graph with the following property. $S$ contains the split vertex $x_e$ for any edge $e$ whose endpoints are both in $S$; since otherwise we can move $x_e$ into $S$ and make $\Phi_{G_E}(S)$ smaller. This property also holds for $A$ and $R$.
		\item For any subset of nodes $U$ in $G_E$ satisfying the property in the previous remark, there is a corresponding subset $U\cap V$ in $G$. Up to a constant factor of $2$, the volume, cut-size and conductance of $U$ in $G_E$ are the same as the volume, cut-size and conductance of $U\cap V$ in $G$ respectively.
	\end{enumerate}
\end{rem}
By the above remarks, although the following discussion is in $G_{E}$,
everything translates easily to $G$.
Now, we analyze the algorithm. Let $A_{t},R_{t},F_{t}$ denote
$A,R,F$ at round $t$ respectively. For each split node $x_{e}$,
let $F_{t}(x_{e})\in R_{\ge0}^{E}$ be the $x_{e}$-row of $F_t$. We
call $F_{t}(x_{e})$ a \emph{flow-vector} of $x_{e}$. We define a
potential function 
\[
\psi(t)=\sum_{x_{e}\in A_{t}}||F_{t}(x_{e})-\mu_{t}||_{2}^{2}
\]
where $\mu_{t}=\sum_{x_{e}\in A_{t}}F_{t}(x_{e})/|X_{E}\cap A_{t}|$
is the average flow vector of split nodes remaining in $A_{t}$.
\begin{lem}
	\label{prop:mixing}For any $t\le T$, if $\psi(t)\le1/16m^{2}$,
	then $A_{t}$ is a $\phi$ nearly expander. \end{lem}
\begin{proof}
	Here we omit the subscript for readability. Let $S$ be any subset of $A$ where $\vol(S)\leq\vol(A)/2$.
	For each $x_{e}\in S$, if $\sum_{x_{h}\in A}F(x_{h},x_{e})\ge1/2$,
	i.e.~the total flow that $x_{e}$ receives from commodities in $A$ is at least $1/2$, then the
	$x_{e}$-coordinate of $\mu_{t}(x_{e})$$\ge1/2m$. As $\psi(t)\le1/16m^{2}$,
	we have that for every $x_{h}\in A-S$, $F(x_{h},x_{e})\ge1/4m$.
	Since $\vol(R)$ is small, and $\vol(S)\leq\vol(A)/2$, we know $\vol(A-S)$
	is $\Omega(m)$, so $x_{e}$ receives a constant fraction at least $\Omega(m)/4m = \Omega(1)$ unit of flow from outside $S$. In the another case, if $\sum_{x_{h}\in A}F(x_{h},x_{e})<1/2$,
	then $x_{e}$ must receive $\Omega(1)$ unit of flow from commodities from $R$. Summing over all $x_{e}\in S$, the multi-commodity
	flow $f$ sends a total amount of at least $\Omega(\vol(S))$ out
	of $S$. Since the congestion of $f$ is at most $ct=t/\phi T\le1/\phi$,
	we have $|E_{G_{E}}(S,V-S)|\geq\Omega(\vol(S)\phi)$. With proper
	adjustment on constants in our parameters, we get $A$ is a nearly
	$\phi$ expander.
\end{proof}

For readers who are familiar with the cut-matching frame-work, the update operation in \ref{alg:cut-matching} just implements (a variant of) one round the cut-matching frame-work with the below guarantee. The proof is shown in \Cref{section:cut-matching-proof}.
\begin{lem}\label{lem:time potential}
	The operation for updating $F$ is each round takes $O(m/(\phi\log m))$ time.
	After $T$ rounds, it holds that $\psi(t)\le1/16m^{2}$ w.h.p. 
\end{lem}

\paragraph{Proof of \ref{thm:cut-matching}.}
Observe first that any round $t$, we have $\Phi_{G_{E}}(R_{t})\le1/c=O(\phi\log^{2}m)$.
This is because $R_{t}=\cup_{1\le t'\le t}$$S_{t}$ and $\Phi_{G_{E}\{A_{t}\}}(S)\le1/c$.
If \ref{alg:cut-matching} terminates because $\vol(R_{t})>m/10T=\Omega(m/\log^{2}m)$,
then $(A_{t},R_{t})$ is a balanced cut where $\Phi_{G_{E}}(R_{t})=O(\phi\log^{2}m)$
and we obtain the second case of \ref{thm:cut-matching}. Otherwise,
\ref{alg:cut-matching} terminates at round $T$ and we apply \Cref{prop:mixing,lem:time potential}.
If $R=\emptyset$, then we obtain the first case
because the whole node set $V\cup X_{E}$ is a nearly $\phi$ expander,
which means that $G_{E}$ is a $\phi$ expander. In the last case,
we write $T=c_{0}\log^{2}m$ for some constant $c_{0}$. We have $\Phi_{G_{E}}(R_{t})\le1/c=c_{0}\phi\log^{2}m$
and $\vol(R_{t})\le m/10T=m/(10c_{0}\log^{2}m)$. Moreover, $A_{T}$
is a nearly $\phi$ expander. 

Again, all these three cases translate to the three cases in \ref{thm:cut-matching}
by \ref{remark:subdivision}.
By \ref{lem:time potential}, the total running time is $T\times O(m/(\phi\log m)) = O(m\log m/\phi)$. This completes the proof of \ref{thm:cut-matching}.

\subsection{Implicitly Updating Multi-commodity Flows}
\label{section:cut-matching-proof}
To prove \Cref{lem:time potential}, we need to describe how to implicitly update $F$. This corresponds to one round in the cut-matching game from \cite{KhandekarRV09,RackeST14}.
In each round $t$, we first find a ``cut'' in $G\{A\}$ using \Cref{lemma:proj,lemma:bisection}. More precisely, we actually find two disjoint sets $A^l,A^r\subset A$ and not a cut.
Then, we run a flow algorithm which tries to routes mass from $A^l$ to $A^r$. The returned flow defines a ``matching'' $M$ between nodes $A^l$ and $A^r$.  
Finally, $M$ indicates how to implicitly update $F$ so that it is still routable with congestion $ct$, and, once $t=T$, $F$ certifies that the remaining $A$ is a nearly $\phi$ expander as  \Cref{lem:time potential} needs.

	Now, we describe the algorithm in details.
Suppose we are now at round $t$. We will update $F_t$ to be $F_{t+1}$. 
%It suffices to show that this takes time $O(m/(\phi \log m))$ and $\phi(t+1) \le \phi(t)(1-1/\Omega(\log m))$.
Let $A^l$ and $A^r$ be constructed by the two lemmas below. 

\begin{lem}[Lemma 3.4 \cite{KhandekarRV09}]
	\label{lemma:proj}
	Let $u_e=\left\langle F_t(x_e),r \right\rangle$ be the projection of $x_e$'s flow-vector to a random unit vector $r$ orthogonal to the all-ones vector. We have $\E[(u_e-u_h)^2]=\norm{F_t(x_e)-F_t(x_h)}_2^2/m$ for all pairs $x_e,x_h$, and 
	\[
	(u_e-u_h)^2\leq \frac{C\log m}{m}\norm{F_t(x_e)-F_t(x_h)}_2^2
	\]
	hold for all pairs with high probability for some constant $C$.
\end{lem} 
Note that $F$ is not explicitly maintained. Despite this fact, \Cref{lemma:proj_time} shows that all $u_e$'s can be computed in $O(m\log^2 m)$.

\begin{lem}[Lemma $3.3$ in \cite{RackeST14}]
	\label{lemma:bisection}
	Given $u_e$'s for all $x_e\in A \cap X_E$, we can find in time $O(|A|\log |A|)$ a set of source split nodes $A^l\subset A$, and a set of target split nodes $A^r\subset A$, and a separation value $\eta$ such that
	\begin{enumerate}
		\item $\eta$ separates the sets $A^l,A^r$, i.e., either $\max_{x_e\in A^l} u_e\leq \eta\leq \min_{x_h\in A^r} u_h$ or $\min_{x_e\in A^l} u_e\geq \eta\geq \max_{x_h\in A^r} u_h$,
		\item $|A^r|\geq |A\cap X_E|/2$, $|A^l|\leq |A\cap X_E|/8$,
		\item $\forall x_e\in A^l: (u_e-\eta)^2\geq \frac{1}{9}(u_e-\bar{\mu})^2$, where $\bar{\mu}=\left\langle \mu_t,r \right\rangle$ is the projection of average flow vector $\mu_t$ on the random direction,
		\item $\sum_{x_e\in A^l}(u_e-\bar{\mu})^2\geq \frac{1}{80}\sum_{x_e\in A}(u_e-\bar{\mu})^2$
	\end{enumerate}
\end{lem}
 
Given $A^l$ and $A^r$, consider a flow problem on $G\{A\}$ where each split nodes in $A^l$ is a source of $1$ unit of mass and each split nodes in $A^r$ is a sink with capacity $1$ unit. Every edge has the same capacity $U = 1/(\phi \log^2 m)$. This flow problem can be easily solved by push-relabel algorithms or blocking-flow algorithms whose labels are bounded by the parameter $h = 1/(\phi \log m)$ (e.g. {\em Unit-Flow} \cite{HenzingerRW17} or the block-flow algorithm in \cite{OrecchiaZ14}). 
\begin{lem}
	\label{lemma:flow-for-matching}
	In time $O(mh) = O(m/(\phi\log m))$, we can either find
	\begin{enumerate}
		\item A feasible flow $f$ for the above flow problem, or
		\item A cut $S$ where $\Phi_{G\{A\}}(S)=O(\phi\log^{2}m)$ and a feasible flow for the above flow problem when only split nodes in $A^l-S$ are sources of 1 unit.
	\end{enumerate}
\end{lem} 

Let $M_t$ be a (non-perfect) matching from nodes in $A^l$ to nodes in $A^r$. If a unit of mass from $x_e \in A^l$ is routed through $f$ to $x_h \in A^r$, then $M_t$ contains an edge $(x_e,x_h)$. We can construct $M_t$ from $f$ using, for example, a dynamic tree \cite{SleatorT83} in $O(m\log m)$ time.
%From $f$, we can 
%We perform the flow path decomposition $f$ into $O(m)$ paths in $O(m\log m)$ time \cite{KangP15}.
%This defines a partial matching from $A^l$ to $A^r$. 
%
From the matching $M_t$, we update $F_t$ to $F_{t+1}$ as follows:
for each $(x_e,x_h) \in M_t$ 
\[
F_{t+1}(x_e)=F_{t}(x_e)/2 + F_{t}(x_h)/2.
\]
Next, we list properties of $F_{t+1}$.
\begin{lem}
	$F_{t+1}$ is routable with congestion $c(t+1)$.\end{lem}
\begin{proof}
	By induction $F_{t}$ is routable with congestion $ct$. As $f$ is
	feasible for the flow problem whose edge capacity is $c$, it has
	congestion $c$. By using the flow path of $f$ to average the commodities,
	we know that $F_{t+1}$ is routable with congestion $ct+c=c(t+1)$.
\end{proof}

%As $F_{t+1}$ can be obtained from $F_t$ in the above simple way, we have:

\begin{lem}
	\label{lemma:proj_time}
	For any $m$-dimensional vector $r$, all $u_e =\left\langle F_t(x_e),r \right\rangle$ can be computed in $O(mt)$ time.
\end{lem} 
\begin{proof}
	Treat $M_t$ has an $m\times m$ adjacency matrix of itself.
	Observe that $F_t = (\frac{I+M_t}{2})\times \cdots \times (\frac{I+M_1}{2})$.
	So $u_e$ is just the $x_e$-th entry of this vector: $F_t r $, which can be computed in $O(mt)$ time because $M_i$ has only $O(m)$ non-zeros for each $i$.
\end{proof}

%This lemma shows that this update reduces the potential:
\begin{lem}[Lemma $3.3$ in \cite{KhandekarRV09}]% or Lemma $3.4$ in \cite{RackeST14} 
	\label{lemma:matching}
	Before removing $S$ from $A$, the potential function $\psi(t)$ is reduced by at least $\frac{1}{2}\sum_{(x_e,x_h)\in M_t} \norm{F_t(x_e)-F_t(x_h)}_2^2$.
\end{lem}

\begin{lem}
	\label{lemma:potential}
	$\psi(t+1)$ is reduced by a $(1-\Omega(1/\log m))$ factor comparing to $\psi(t)$ with high probability.
\end{lem}
\begin{proof}
	Let us first define the following quantities for tracking the change
	of potential in a more fine-grained way:
	\begin{itemize}
		\item $\psi(t)=\sum_{x_{e}\in A_{t}}\|F_{t}(x_{e})-\mu_{t}\|_{2}^{2}$ is
		the potential at the beginning of round $t$.
		\item $\psi_{match}(t)=\sum_{x_{e}\in A_{t}}\|F_{t+1}(x_{e})-\mu_{t}\|_{2}^{2}$
		is the potential after updating the flow vectors via the matching
		$M_{t}$.
		\item $\psi_{delete}(t)=\sum_{x_{e}\in A_{t+1}}\|F_{t+1}(x_{e})-\mu_{t}\|_{2}^{2}$
		is the potential after deleting the flow vectors on $S$. Note that
		we sum over $x_{e}\in A_{t+1}$ because $A_{t+1}=A_{t}-S$.
		\item $\psi(t+1)=\sum_{x_{e}\in A_{t+1}}\|F_{t+1}(x_{e})-\mu_{t+1}\|_{2}^{2}$
		is the potential at the beginning of round $t+1$. That is, we replace
		$\mu_{t}$ with $\mu_{t+1}$ in $\psi_{delete}(t)$.
	\end{itemize}
	First, we have the following:
	
	\begin{align*}
		& \psi(t)-\psi(t+1)\\
		& =\left(\psi(t)-\psi_{match}(t)\right)+\left(\psi_{match}(t)-\psi_{delete}(t)\right)+\left(\psi_{delete}(t)-\psi(t+1)\right)\\
		& =\underset{\text{\ref{lemma:matching}}}{\underbrace{\left(\frac{1}{2}\sum_{(x_{e},x_{h})\in M_{t}}\|F_{t}(x_{e})-F_{t}(x_{h})\|_{2}^{2}\right)}}+\underset{\text{Note: }F_{t+1}(x_{e})=F_{t}(x_{e})\forall x_{e}\notin M_{t}}{\underbrace{\left(\sum_{x_{e}\in S}\|F_{t+1}(x_{e})-\mu_{t}\|_{2}^{2}\right)}}\\
		& +\underset{\ge0}{\underbrace{\left(\sum_{x_{e}\in A_{t+1}}\|F_{t+1}(x_{e})-\mu_{t}\|_{2}^{2}-\sum_{x_{e}\in A_{t+1}}\|F_{t+1}(x_{e})-\mu_{t+1}\|_{2}^{2}\right)}}\\
		& \ge\frac{1}{2}\sum_{(x_{e},x_{h})\in M_{t}}\|F_{t}(x_{e})-F_{t}(x_{h})\|_{2}^{2}+\sum_{x_{e}\in S}\|F_{t}(x_{e})-\mu_{t}\|_{2}^{2}
	\end{align*}
	Note that $\sum_{x_{e}\in A_{t+1}}\|F_{t+1}(x_{e})-\mu_{t}\|_{2}^{2}\ge\sum_{x_{e}\in A_{t+1}}\|F_{t+1}(x_{e})-\mu_{t+1}\|_{2}^{2}$
	because $\mu_{t+1}$ is the \emph{mean} of $F_{t+1}(x_{e})$ overall
	$x_{e}\in A_{t+1}$, which is when we exploit the geometric property
	of the mean. Now, we can conclude
%	Recall that $\psi(t)=\sum_{x_{e}\in A}||F_{t}(x_{e})-\mu_{t}||_{2}^{2}$. 
	\begin{align*}
	& \psi(t)-\psi(t+1)\\ 
	\geq & \sum_{x_{e}\in S}\norm{F_{t}(x_{e})-\mu_{t}}_{2}^{2}+\frac{1}{2}\sum_{(x_{e},x_{h})\in M_{t}}\norm{F_{t}(x_{e})-F_{t}(x_{h})}_{2}^{2}\\
	\ge & \frac{m}{C\log m}\sum_{x_{e}\in A^{l}\cap S}(u_{e}-\bar{\mu})^{2}+\frac{m}{2C\log m}\sum_{(x_{e},x_{h})\in M_{t}}(u_{e}-u_{h})^{2}\\
	\geq & \frac{m}{C\log m}\sum_{x_{e}\in A^{l}\cap S}(u_{e}-\bar{\mu})^{2}+\frac{m}{2C\log m}\sum_{x_{e}\in A^{l}\cap(A-S)}(u_{e}-\eta)^{2}\\
	\geq & \frac{m}{C\log m}\sum_{x_{e}\in A^{l}\cap S}(u_{e}-\bar{\mu})^{2}+\frac{m}{18C\log m}\sum_{x_{e}\in A^{l}\cap(A-S)}(u_{e}-\bar{\mu})^{2}\\
	\geq& \sum_{x_{e}\in A^{l}}\frac{m}{18C\log m}(u_{e}-\bar{\mu})^{2}\\
	\geq& \sum_{x_{e}\in A}\frac{m}{1440C\log m}(u_{e}-\bar{\mu})^{2}
	\end{align*}
	
%	In the first inequality, the first term is from removing $S$ from $A$ in round $t+1$ and the second term is by \ref{lemma:matching}. 
	The second inequality follows from \ref{lemma:proj}. The third inequality follows by \ref{lemma:bisection}$(1)$. The fourth and the last inequalities follow \ref{lemma:bisection}$(3)$ and $(4)$ respectively. The above bound hold with high probability. Now use \ref{lemma:proj} again, we know the expectation of the potential drop is
	\begin{align*}
	\E[\sum_{x_e\in A}\frac{m}{1440C\log m}(u_e-\bar{\mu})^2] & =\sum_{x_e\in A}\frac{\norm{f_t(x_e)-\mu_t}_2^2}{1440C\log m} \\
	& = \Omega(\psi(t)/\log m)
	\end{align*}
\end{proof}

\end{appendix}

\end{document}